\documentclass[10pt,final,journal]{IEEEtran} 

\IEEEoverridecommandlockouts % The preceding line is only needed to identify funding in the first footnote. If that is unneeded, please comment it out.

\usepackage{def/mypaper}

\makeatother
\makeatletter
\newcommand*{\tsp}{%
	{\mathpalette\@tsp{}}%
}
\newcommand*{\@tsp}[2]{%
	\raisebox{\depth}{$\m@th#1\intercal$}%
}
\makeatother

\begin{document}

\title{Decentralized State Estimation In A Dimension-Reduced Linear Regression 
%\thanks{This work has been supported by the Industry Excellence Center LINK-SIC funded by The Swedish Governmental Agency for Innovation Systems (VINNOVA) and Saab AB, and by the project Scalable Kalman filters funded by the Swedish Research Council (VR). G.~Hendeby has received funding from the Center for Industrial Information Technology at Link\"{o}ping University (CENIIT) grant no. 17.12.}
}

\author{\IEEEauthorblockN{Robin Forsling\IEEEauthorrefmark{1}, Fredrik Gustafsson\IEEEauthorrefmark{3}, Zoran Sjanic, and Gustaf Hendeby\IEEEauthorrefmark{2}}\\
\IEEEauthorblockA{\IEEEauthorrefmark{1} Student~Member, IEEE, \IEEEauthorrefmark{2} Senior~Member, IEEE, \IEEEauthorrefmark{3} Fellow, IEEE \\
\textit{Dept.\ of Electrical Engineering, Linköping University}, Linköping, Sweden}
}

\maketitle

% ABSTRACT
\begin{abstract}
Decentralized state estimation in a communication-constrained sensor network is considered. The exchanged estimates are dimension-reduced to reduce the communication load using a linear mapping to a lower-dimensional space. The mean squared error optimal linear mapping depends on the particular estimation method used. Several dimension-reducing algorithms are proposed, where each algorithm corresponds to a commonly applied decentralized estimation method. All except one of the algorithms are shown to be optimal. For the remaining algorithm, we provide a convergence analysis where it is theoretically shown that this algorithm converges to a stationary point and numerically shown that the convergence rate is fast. A message-encoding solution is proposed that allows for efficient communication when using the proposed dimension reduction techniques. We also derive different properties from the proposed framework and show its superiority in relation to baseline methods. The applicability of the different algorithms is demonstrated using a simple fusion example and a more realistic target tracking scenario.
\end{abstract}


% SECTIONS

% --- INTRODUCTION ---
\section{Introduction} \label{sec:intro}

Decentralized estimation has been studied extensively over the last few decades. The primary goal is to improve estimates in a \emph{sensor network} (\abbrSN) while maintaining the benefits of a decentralized design: \emph{robustness} and \emph{modularity} \cite{Uhlmann2003IF}. A disadvantage of a decentralized \abbrSN is that in general only suboptimal estimates \cite{Forsling2022TSP} can be derived. This is in contrast to a \emph{centralized} \abbrSN, where it is possible to obtain \emph{mean squared error} (\abbrMSE) optimal estimates. However, centralized SNs have higher communication requirements because a central agent processes all of the information that the sensors extract. Despite the fact that the communication cost is reduced in decentralized SNs due to local preprocessing of information, decentralized SNs also impose communication constraints \cite{Kimura2005ITCC}. Developing communication-reducing techniques is therefore crucial in general network-centric problems \cite{Chen2022TSIPN,Liu2023TSIPN}.

\begin{figure}[t]
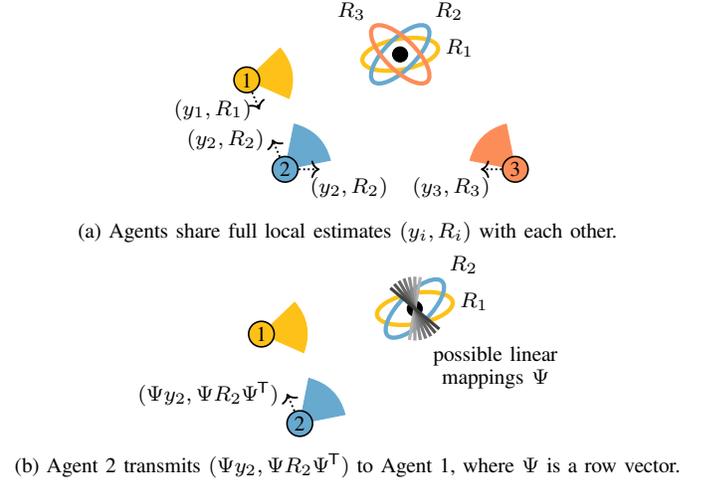

	\centering
	\begin{subfigure}[b]{1.0\columnwidth}
		\centering
		\begin{tikzpicture}[scale=.17]
			\input{fig/example_intro_full.tex}
		\end{tikzpicture}
		\caption{Agents share full local estimates $(y_i,R_i)$ with each other.}
		\label{fig:intro:example-full}
	\end{subfigure}
	\begin{subfigure}[b]{1.0\columnwidth}
		\centering
		\begin{tikzpicture}[scale=.17]
			\input{fig/example_intro_reduced.tex}
		\end{tikzpicture}
		\caption{Agent~2 transmits $(\M y_2,\M R_2\Mt)$ to Agent~1, where $\M$ is a row vector. } % The question is then how to optimally choose $\M$. In this case $\M$ projects $y_2$ onto a 1D subspace with multiple such subspaces represented by lines in the figure. % Similar setting as in (a) but now with the focus only on agent~1 and agent~2. s
		\label{fig:intro:example-reduced}
	\end{subfigure}
	\caption{A decentralized sensor network where the target is given by a black circle. Colored circles and cones resemble agents and sensors, respectively. Dashed arrows illustrate communication links. } % $\Ds_2\succeq R_2$
	\label{fig:intro:example}
\end{figure}

A decentralized \abbrSN is illustrated in Fig.~\ref{fig:intro:example-full}. Multiple agents use local measurements to derive local estimates of a common target state $x$, \ie, $(y_i,R_i)$, where $y_i$ is the local estimate of the Agent~$i$ and $R_i$ is the covariance of $y_i$. Local estimates are exchanged between the agents for further improvement. To reduce the communication load, an agent, \eg, Agent~2, can instead exchange a linear combination of $y_2$, $\M y_2$, with covariance $\M R_2\Mt$, where $\M$ is a wide matrix; see Fig.~\ref{fig:intro:example-reduced}. The problem is then how to best design $\M$.

In this paper, we consider a decentralized \abbrSN where the task is to compute optimal estimates in a communication-constrained \abbrSN. To cope with the communication constraints, only \emph{dimension-reduced} (\abbrDR) estimates are exchanged. The dimension reduction is accomplished by the matrix $\M$ which hence determines the overall estimation performance.

The contributions are\footnote{An edited version of this manuscript is included in Chapter~5 of the first author's PhD thesis \cite{Forsling2023Phd}.}:
\begin{itemize}
	\item We extend the previous work in \cite{Forsling2022Fusion} from row vector valued $\M$ to matrix valued $\M$.
	\item A convergence analysis is provided for the algorithm developed for computing $\M$ in the \emph{covariance intersection} (\abbrCI, \cite{Julier1997ACC}) case. %Moreover, a numerical convergence rate analysis is given.
	\item The methodology used to select $\M$ is extended to the \emph{largest ellipsoid} (\abbrLE, \cite{Benaskeur2002IECON}) method. We formulate conditions for when the derived $\M$ is optimal.
	\item A message encoding algorithm is proposed that encodes all data contained in $(\M y_2,\M R_2\Mt,\M)$ for high communication efficiency. 
	\item We discuss and derive results dedicated to the singular case and how to assign the number of rows in $\M$.
\end{itemize}

% --- BACKGROUND ---
\section{Background} \label{sec:background}

In network-centric problems, multiple agents are cooperating for a common goal, \eg, estimating a common target. The focus is on estimation in decentralized SNs, where an agent consists of sensors and a processing unit. We begin by describing different network aspects. Prior work is reviewed at the end of this section.

% NETWORK-CENTRIC ESTIMATION
\subsection{Network-Centric Estimation} \label{sec:network-centric-estimation}

Network-centric estimation is a well-studied subject \cite{Liggins2009-ch17}. By utilizing multiple sensors and the communication of information, it is possible to enhance the quality of estimates. This information is then exchanged for further refinement. A centralized \abbrSN is illustrated in Fig.~\ref{fig:bkg:csn}, where sensor measurements are communicated to a central processing unit where estimates are derived. Two major drawbacks of the centralized \abbrSN include vulnerability to the failure of critical nodes and a high communication load since all measurements are transmitted to the central node. 

%In a distributed \abbrSN local measurements are preprocessed before it is distributed among the other agents. While enhancing robustness and scalability, the distributed \abbrSN has certain inherit assumptions about the network topology, state-space models, and local preprocessing schemes. For instance, the distributed Kalman filter algorithms of \cite{Khan2008TSP,Govaers2010Fusion,Reinhardt2012MFI} assume that the implemented filters follow predefined schemes, and the diffusion-based strategies of \cite{Lopes2008TSP,Cattivelli2010TSP} are developed under certain assumptions about the network topology.

\begin{figure}[tb]
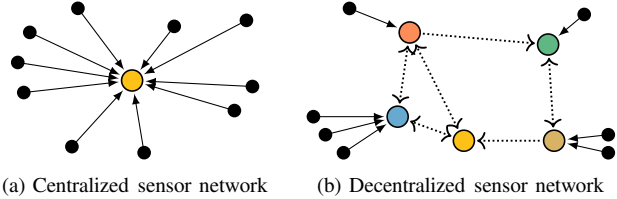

	\centering
	\begin{subfigure}[b]{0.475\columnwidth}
		\centering
		\begin{tikzpicture}[scale=.08]
			\input{fig/centralized_sensor_network.tex}
		\end{tikzpicture}
		\caption{Centralized sensor network}
		\label{fig:bkg:csn}
	\end{subfigure}
	\begin{subfigure}[b]{0.475\columnwidth}
		\centering
		\begin{tikzpicture}[scale=.08]
			\input{fig/decentralized_sensor_network.tex}
		\end{tikzpicture}
		\caption{Decentralized sensor network}
		\label{fig:bkg:dsn}
	\end{subfigure}
	\caption{Sensor networks. Black and colored circles resemble sensor nodes and processing units, respectively. The solid and dotted arrows illustrate sensor-to-processing unit and inter-processing unit communication, respectively.} 
	\label{fig:bkg:sn}
\end{figure}

This paper considers decentralized SNs. A decentralized SN is illustrated in Fig.~\ref{fig:bkg:dsn}. This type of network is, \eg, realized by wireless ad hoc networks such as wireless SNs \cite{Toh1997WANET}. According to \cite{Grime1994CEP}, a decentralized SN is characterized by the following properties:
\begin{enumerate}[(i)]
	\item There is no single point of failure.
	\item Agents communicate on a agent-to-agent basis without a central communication management.
	\item There is no global knowledge about the network topology.
\end{enumerate}
In a decentralized \abbrSN, measurements are passed to a local processing unit, where local estimates are computed. The local estimates are then exchanged between the processing units to enhance the tracking quality through fusion. A decentralized \abbrSN has the following advantages over its centralized counterpart \cite{Julier2009-ch14,Uhlmann2003IF}:
\begin{itemize}
	\item \emph{Robustness}. Since there is no single point of failure, a decentralized design is inherently fault-tolerant. 
	\item \emph{Modularity}. The network is decomposed into smaller, self-contained subsystems, which reduces the overall complexity and makes the network scalable.
	\item \emph{Flexibility}. Agents can connect to and disconnect from the network on-the-fly.
\end{itemize}

A decentralized \abbrSN is an example of distributed processing. However, the terms decentralized and distributed should not be used interchangeably. While distributed sensor networks have multiple interconnected processing units, there is typically a central coordinator or some global processing result is obtained, \eg, globally available estimates \cite{Castanedo2013TSWJ}. This work is limited to methods that are able to cope with the constraints imposed by (i)--(iii), thereby excluding many distributed processing techniques. For instance, the distributed Kalman filter algorithms of \cite{Khan2008TSP,Govaers2010Fusion} assume that the implemented filters follow predefined schemes, and the diffusion-based strategies of \cite{Lopes2008TSP,Cattivelli2010TSP} are developed under certain assumptions about the network topology.

% PRIOR WORK
\subsection{Prior Work} \label{sec:prior-work}

A crucial aspect of network-centric estimation problems is communication constraints \cite{Razzaque2013ACM}. Estimates and other target information need to be communicated to a legion of targets and to and from many agents. At some point, the communication link will become a bottleneck, and hence the finite size of the communication resource must be taken into account. There are also situations where the communication needs to be reduced for other reasons, \eg, to be able to operate with a low electromagnetic signature. Hence, there is a need to study DTT under communication constraints \cite{Kimura2005ITCC}.

% ADD PARAGRAPH ABOUT PCA AND LOW-RANK STUFF

In \cite{Zhang2003Fusion}, it is shown that the problem of finding an optimal linear mapping that compresses measurements to be exchanged boils down to an eigenvalue problem. In \cite{Chen2004CDC}, the problem is extended to the dynamic case. Distributed estimation based on \abbrDR data is handled in \cite{Zhu2005TSP} and \cite{Schizas2007TSP}. In the latter, a non-ideal communication channel is also considered. The authors of \cite{Fang2010TSP} address the problem of jointly assigning the dimensionality and deriving optimal data compression matrices. Performance bounds of data compression techniques are analyzed in \cite{Ma2014TSP}. A framework for dimension reduction, including data denoising for distributed algorithms, is proposed in \cite{Schizas2007TSP}. There is numerous additional work on closely related dimension-reduction problems; see, \eg, \cite{Schizas2015TSP,Fang2008SPL,Fang2011DSP}. 

 %The authors of \cite{Greiff2020ACC,Greiff2020CCTA} use an gradient descent algorithm to optimize the dimension-reduction.

An alternative strategy for reducing the communication load is to quantize the data to be exchanged; see, \eg, \cite{Ribeiro2006TSP,Msechu2012TSP}. A key aspect of quantization in \abbrDSN is how to preserve conservativeness when the communicated data is quantized \cite{Funk2021Sensors}. The authors of \cite{Kar2009ICASSP,Kar2012TIT} deals with network-centric estimation under quantized and imperfect communications. In this paper, however, it is assumed that quantization effects can be neglected.

% \cite{Msechu2012TSP,Venk2007TSP,Ribeiro2006TSP}

% --- PROBLEM ---
\section{Fusion of Dimension-Reduced Estimates} \label{sec:problem}

This paper considers a decentralized \abbrSN, where only \abbrDR estimates are allowed to be exchanged. The task is to reduce dimensionality so that fusion performance is promoted.

Mathematical preliminaries are given at first. Then the problem is formalized. Several methods applicable for the fusion of \abbrDR estimates are provided. At the end, we define two baseline methods for communication reduction, with which we will compare the developed methods.

% PRELIMINARIES
\subsection{Preliminaries} \label{sec:preliminaries}

Let $\reals$, $\realsn$, and $\realsmn$ denote the sets of all real numbers, $n$-dimensional real-valued vectors, and $m\times n$ real-valued matrices, respectively. Denote by $\psdsetn$ and $\pdsetn$ the sets of all $n\times n$ symmetric positive semidefinite (\abbrPSD) matrices and all $n\times n$ symmetric positive definite (\abbrPD) matrices. For $A,B\in\psdsetn$, the inequalities $A\succeq B$ and $A\succeq B$ are equivalent to $(A-B)\in\psdsetn$ and $(A-B)\in\pdsetn$, respectively. The expected value and the covariance of $a$ are given by $\EV(a)$ and $\cov(a)=\EV\left((a-\EV(a))(a-\EV(a))\trnsp\right)$, respectively. The cross-covariance of $a$ and $b$ is given by $\cov(a,b)=\EV\left((a-\EV(a))(b-\EV(b))\trnsp\right)$. The ellipsoid of $A\in\pdsetn$ is defined by the set of points $\calE(A)=\{ a\in\realsn \,|\, a\trnsp A\inv a = 1\}$. By $\rowspan(A)$, we denote the row span of $A$.

The state to be estimated is denoted by $x\in\realsnx$. It is assumed that\footnote{The suggested model implies that we adopt a \emph{Fisherian} view instead of the \emph{Bayesian} view. However, in the context of a Bayesian framework, the estimate $(y_1,R_1)$ can be interpreted as a prior.}
\begin{subequations} \label{eq:estimate-model}
\begin{align}
	y_1 &= x + v_1, & R_1 &= \cov(v_1), \\
	y_2 &= H_2x + v_2, & R_2 &= \cov(v_2),
\end{align}	
\end{subequations}
where $y_2\in\realsnb$, $H_2\in\realsnbnx$, and $R_{12}=\cov(v_1,v_2)$. Subscripts 1 and 2 refer to Agent~1 and Agent~2. In this context, $y_1$ and $y_2$ are considered to be estimates of $x$ and $H_2x$, respectively. However, all results hold if $y_1$ and $y_2$ are, \eg, measurements. The estimation error is given by $\xtilde=\xhat-x$, and $P$ is the covariance associated with $x$. An estimate $\xhat$ is \abbrMSE optimal if it minimizes $\trace(P)$. It is \emph{conservative} if
\begin{equation}
	P - \EV(\xtilde\xtilde\trnsp) \succeq 0.
	\label{eq:conservative-estimate}
\end{equation}

A generalized eigenvalue $\lambda(A,B)$ of $A,B$ is given by
\begin{equation}
	Au = \lambda Bu,
	\label{eq:gevp}
\end{equation}
where $u$ is a generalized eigenvector associated with $\lambda$. If $A,B\in\pdsetn$, then the \emph{generalized eigenvalue problem} (\abbrGEVP) in \eqref{eq:gevp} has $n$ solutions $(\lambda,u)$, not necessarily unique. The \ith generalized eigenvalue and eigenvector are denoted $\lambda_i$ and $u_i$, respectively. If $B=I$, then the \abbrGEVP in \eqref{eq:gevp} reduces to an ordinary \emph{eigenvalue problem} (\abbrEVP).

If $A\in\pdsetn$ it has en eigendecomposition
\begin{equation}
	A = U\Sigma U\trnsp = \sum_{i=1}^n \lambda u_iu_i\trnsp,
	\label{eq:eigendecomposition}
\end{equation}
where $U$ is an orthogonal matrix such that $u_i\trnsp u_j=\delta_{ij}$, where $\delta_{ij}$ is the Kronecker delta. The matrix $\Sigma\in\pdsetn$ is diagonal, with the eigenvalues of $A$ on its diagonal. (Generalized) eigenvalues are assumed to be given in descending order
\begin{equation}
	\lambdamax = \lambda_1 \geq \dots \geq \lambda_n = \lambdamin.
	\label{eq:ev-ordering}
\end{equation}

% PROBLEM STATEMENT
\subsection{Problem Statement}

Let $\M\in\realsmnb$, where $m<\nb$ and $\rank(\M)=m$. A \abbrDR estimate is given by $(\yM,\RM)$ constructed as
\begin{align}
	\yM &= \M y_2, & \RM &= \M R_2\Mt.
	\label{eq:dr-esimate}
\end{align} 
Let $m$ be fixed and assume that Agent~2 is about to transmit $(\yM,\RM)$ to Agent~1. The problem is to compute $\M$ such that when Agent~1 computes $(\xhat,P)$ by fusing $(y_1,R_1)$ and $(\yM,\RM)$, $\trace(P)$ is minimized. The considered fusion methods are defined in the next.

% FUSION METHODS
\subsection{Fusion Methods}

The \emph{Bar-Shalom-Campo} (\abbrBSC, \cite{Bar-Shalom1986TAES}) fuser, the \emph{Kalman fuser} (\abbrKF), \abbrCI, and the \abbrLE method are defined below for fusion of $(y_1,R_1)$ and $(\yM,\RM)$. These are all linear unbiased fusion methods \cite{Forsling2020Lic}. In case of \abbrBSC, \abbrKF, and \abbrCI, $\xhat$ is given as
\begin{equation}
	\xhat = K_1y_1 + \KM\yM = (I-\KM\M H_2)y_1 + \KM\yM,
	\label{eq:dr:linear-fusion}
\end{equation}
where $K_1=I-\KM\M H_2$ is follows from the unbiasedness constraint $\BBSM K_1&\KM\EBSM\BBSM I\\ \M H_2 \EBSM=I$. This means that $(\xhat,P)$ is fully specified by $(\KM,P)$.

% BSC
\subsubsection{Bar-Shalom-Campo Fusion}

\abbrBSC for the fusion of two correlated estimates is originally presented in \cite{Bar-Shalom1986TAES} in the special case $H_2=I$. A more general version of \abbrBSC is derived in, \eg, \cite{Forsling2023Phd} for arbitrary $H_2$. The estimate $(\xhat,P)$ is specified by
\begin{align}
	\KM &= (R_1H_2\trnsp-R_{12})\Mt\SM\inv, &
	P &= R_1 - \KM\SM\KM\trnsp,
	\label{eq:method:bsc-dr}
\end{align}
where $\SM=\M H_2R_1H_2\trnsp\Mt+\RM-\M H_2R_{1\M}-R_{1\M}\trnsp H_2\trnsp\Mt$. If $R=\BBSM R_1&R_{12}\\R_{21}&R_2\EBSM\succ0$, then $\SM\succ0$. If $R=\BBSM R_1&R_{12}\\R_{21}&R_2\EBSM\succeq0$, then $\SM\succeq0$ might be singular. In this case, the pseudoinverse $\SM\pinv$ can be used instead of $\SM\inv$ to compute $\KM$ in \eqref{eq:method:bsc-dr}.

% KF
\subsubsection{Kalman Fusion}

The \abbrDR version of \abbrKF is recovered by setting $R_{12}=0$ in \eqref{eq:method:bsc-dr}. In this case the estimate is specified by
\begin{align}
	\KM &= R_1H_2\trnsp\Mt\SM\inv, &
	P &= R_1 - \KM\SM\KM\trnsp,
	\label{eq:method:kf-dr}
\end{align}
where $\SM=\M H_2R_1H_2\trnsp\Mt+\RM$.

% CI
\subsubsection{Covariance Intersection}

A general version of \abbrCI is defined in \cite{Uhlmann2003IF}. Given this it is possible to define
\begin{subequations} \label{eq:method:ci-dr}
\begin{align}
	\KM &= (1-\omega)P H_2\trnsp\Mt\RM\inv, \\
	P &= (\omega R_1\inv + (1-\omega)H_2\trnsp\Mt\RM\inv\M H_2 )\inv,
\end{align}		
\end{subequations}
where $\omega\in(0,1]$. Given that $(y_1,R_1)$ and $(\yM,\RM)$ are conservative, it has been shown that \abbrCI produces conservative estimates for all admissible $\omega$ \cite{Julier1997ACC}. If $R_{12}$ is completely unknown, then \abbrCI is an optimal conservative estimator \cite{Reinhardt2015SPL}.

% LE
\subsubsection{The Largest Ellipsoid Method}

The original version of \abbrLE, applicable for $H_2=I$, is derived in \cite{Benaskeur2002IECON}. A generalization of \abbrLE for arbitrary $H_2$ is proposed in \cite{Forsling2020Lic}. This version can be adapted for \abbrDR estimates by replacing $y_2$, $R_2$, and $H_2$ by $\yM$, $\RM$, and $\M H_2$, respectively. In this case, $(\xhat,P)$ is computed according to
\begin{align}
	\xhat &= PT\inv\iota, &
	P &= (T\inv\calI T\invtrnsp)\inv, 
	\label{eq:method:le-dr:estimate}
\end{align}	
where $T=T_2T_1$,
\begin{subequations} \label{eq:method:le-dr:transformed-domain}
\begin{align*}
	\iota_1 &= R_1\inv y_1, & \calI_1 &= R_1\inv, \\
	\iota\subM &= H_2\trnsp\Mt\RM\inv\yM, & \calI\subM &= H_2\trnsp\Mt\RM\inv\M H_2, \\
	\calI_1 &= U_1\Sigma_1 U_1\trnsp, & T_1 &= \Sigma_1^{-\frac{1}{2}}U_1\trnsp, \\
	T_1\calI\subM T_1\trnsp &= U\subM\Sigma\subM U\subM\trnsp, & T_2 &= U\subM\trnsp, \\
	\iota_1' &= T\iota_1, & \calI_1' &= T\calI_1T\trnsp = I, \\
	\iota\subM' &= T\iota\subM, & \calI\subM' &= T\calI\subM T\trnsp,
\end{align*}	
\end{subequations}
and
\begin{equation}
	([\iota]_i,[\calI]_{ii}) = 
	\begin{cases}
		([\iota_1']_i,[\calI_1']_{ii})	, &\text{ if } 1 \geq [\calI\subM']_{ii}, \\
		([\iota\subM']_i,[\calI\subM']_{ii})	, &\text{ if } 1 < [\calI\subM']_{ii},
	\end{cases}
	\label{eq:method:le-dr:transformed-estimate}
\end{equation}
for each $i=1,\dots,\nx$.

\begin{rmk}
By property (ii) of Sec.~\ref{sec:network-centric-estimation} the agents communicate on an agent-to-agent basis only. Hence, there is no restriction in assuming that Agent~2 communicates with Agent~1. Since all of the previously described fusion methods can be applied sequentially this is true even when multiple agents simultenously transmit estimates to Agent~1. However, sequential pairwise optimal fusion in general leads suboptimal results compared to batch-wise optimal fusion.
\end{rmk}

% BASELINE METHODS
\subsection{Baseline Methods for Communication Reduction}

The proposed methodology for fusion optimal dimension reduction is derived in Sec.~\ref{sec:gevo}. This methodology is compared to the following two techniques for communication reduction.

% PCO
\subsubsection{The Principal Component Optimization Method} \label{sec:pco}

Popular methods for \abbrDR in general are the \emph{principal component analysis} (\abbrPCA, \cite{Pearson1901PCA}) and the closely related \emph{Karhunen-Lo\`{e}ve transform} \cite{Wang2012KLT}. Both of these approaches utilize eigendecomposition. It should be pointed out that there are many other closely related approaches for \abbrDR, and these concepts are sometimes used interchangeably, \eg, low-rank strategies \cite{Scharf1987TASSP,Honig2002TCOM,DeLamare2007TSP}. 
% KLT: Amar2010TSP 
% Low-rank: Hua2001TSP

Here, the \abbrPCA is utilized by letting $\M$ to be defined by the eigenvectors $u_{\nb-m+1},\dots,u_{\nb}$ corresponding to the $m$ smallest eigenvalues of $R_2$. The resulting method is referred to as the \emph{principal component optimization} (\abbrPCO) method and is provided in Algorithm~\ref{alg:pco}. The mapping $\M$ computed in this way solves the \abbrPCA problem
\begin{equation}
	\begin{aligned}
		& \underset{\M}{\minimize} & & \trace( \M R_2\Mt ) \\
		& \subjectto & & \M\Mt = I.
	\end{aligned}
\end{equation}

\begin{algorithm}[tb]
	\caption{\abbrPCO}
	\label{alg:pco}
	\begin{footnotesize}
	\begin{algorithmic}[0]
		\Input $R_2\in\pdsetnb$ and $m\leq n_2$
		\begin{enumerate}
			\item Let $R_2=\sum_{i=1}^{n_2}\lambda_i u_iu_i\trnsp$, where $\lambda_1\geq\dots\geq\lambda_{n_2}$ and $u_i\trnsp u_j=\delta_{ij}$.
			\item Define $\M=\col(u_{\nb-m+1}\trnsp,\dots,u_{\nb}\trnsp)$.
		\end{enumerate}
		\Output $\M$
	\end{algorithmic}
	\end{footnotesize}
\end{algorithm}

% DCA 
\subsubsection{The Diagonal Covariance Approximation} \label{sec:dca}

The \emph{diagonal covariance approximation} (\abbrDCA) is proposed in \cite{Forsling2019Fusion} as a framework for reduced communication load. Most of the \abbrDCA methods are based on Agent~2 transmitting $(y_2,\Ds_2)$ to Agent~1, where $\Ds_2\succeq R_2$ is diagonal\footnote{The condition $\Ds_2\succeq R_2$ is imposed to preserve conservativeness.}. In the current scope, we will only use the \emph{eigenvalue based scaling} method (DCA-EIG) defined in Algorithm~\ref{alg:dca-eig}. In \cite{Forsling2023Phd} it is shown that $s$ computed as in Algorithm~\ref{alg:dca-eig} is the smallest $s$ such that $sD_2\succeq R_2$.

\begin{algorithm}[tb]
	\caption{DCA-EIG}
	\label{alg:dca-eig}
	\begin{footnotesize}
	\begin{algorithmic}[0]
		\Input $R_2\in\pdsetnb$ 
		\begin{enumerate}
			\item Let $D_2\in\pdsetnb$ be a diagonal matrix, where $[D_2]_{ii}=[R_2]_{ii}$.
			\item Define $\Ds_2=s D_2$, where $s=\lambdamax(D_2\invsqrt R_2D_2\invsqrt)$.
		\end{enumerate}
		\Output $\Ds_2$
	\end{algorithmic}
	\end{footnotesize}
\end{algorithm}

% --- FUSION OPTIMAL DR ---
\section{Fusion Optimal Dimension-Reduction} \label{sec:fusion-optimal}

The core problem is to find a matrix $\M$ that is optimal with respect to (\wrt) the fusion of $(y_1,R_1)$ and $(\yM,\RM)$, where $\yM=\M y_2$ and $\RM=\M R_2\Mt$. The fusion optimal $\M$ depends on the particular fusion method used. Fusion optimal $\M$ is derived for the \abbrBSC, \abbrKF, and \abbrLE. In the \abbrCI case, we provide a convergence analysis. The section is concluded with a theoretical comparison between the proposed method and \abbrPCO.

% GEVO
\subsection{The Generalized Eigenvalue Optimization Method} \label{sec:gevo}

The developed framework for fusion optimal dimension reduction is based on minimizing $\trace(P)$, where $P$ is according to the \abbrBSC fuser in \eqref{eq:method:bsc-dr}. Consider $P$ given in \eqref{eq:method:bsc-dr}, \ie, 
\begin{equation}
	P = R_1 - \KM \M S\Mt \KM\trnsp, 
	\label{eq:gevo:P-to-minimize}
\end{equation}
where 
\begin{align*}
	\KM &= (R_1H_2\trnsp-R_{12})\Mt(\M S\Mt)\inv, \\
	S &= H_2R_1H_2\trnsp + R_2 - H_2R_{12} - R_{12}\trnsp H_2\trnsp.
\end{align*}
We seek $\M\in\realsmnb$ that solves the problem\footnote{Similar formulations are considered in, \eg, \cite{Zhang2003Fusion,Zhu2005TSP}.}
\begin{equation}
	\begin{aligned}
		& \underset{\M}{\minimize} & & \trace(P),
	\end{aligned}	
	\label{eq:opt:dr:gevo}
\end{equation}
where $P$ is according to \eqref{eq:gevo:P-to-minimize}. The solution to \eqref{eq:opt:dr:gevo} is denoted $\Mopt$ and is derived in the next. 

Assume $S\succ0$. The degenerate case where $S\succeq0$ might be singular is discussed later on. Let $\Delta=R_1H_2\trnsp-R_{12}$ such that $P$ in \eqref{eq:gevo:P-to-minimize} is given by
\begin{equation}
	P = R_1 - \Delta\Mt(\M S\Mt)\inv\M\Delta\trnsp.
	\label{eq:gevo:P-to-minimize-2}
\end{equation}
Since $R_1$ is constant, minimization of $\trace(P)$ is equivalent to
\begin{equation}
	\begin{aligned}
		& \underset{\M}{\maximize} & & \trace\left( \Delta\Mt(\M S\Mt)\inv\M\Delta\trnsp \right).
	\end{aligned}
	\label{eq:opt:dr:gevo-max}
\end{equation}
Using the cyclic property of trace, this problem is equivalently expressed as
\begin{equation}
	\begin{aligned}
		& \underset{\M}{\maximize} & & \trace\left( (\M S\Mt)\inv \M Q\Mt\right),
	\end{aligned}
	\label{eq:opt:dr:gevo-max-2}	
\end{equation}
where $Q=\Delta\trnsp\Delta$. Since $S\in\pdsetnb$ and $\rank(\M)=m$, it follows that $\M S\Mt\in\pdsetm$. This implies that there exists an invertible matrix $T$ such that $T\M S\Mt T\trnsp=I$. Hence, without loss of generality (\wolog) it can be assumed that $\M S\Mt=I$. Moreover, $S\in\pdsetnb$ implies that $S=LL\trnsp$, where $L$ is invertible. Define $\Phi=\M L$. Then the following problem is equivalent to \eqref{eq:opt:dr:gevo-max-2}
\begin{equation}
	\begin{aligned}
		& \underset{\Phi}{\maximize} & & \trace( \Phi A\Phi\trnsp ) \\
		& \subjectto & & \Phi\Phi\trnsp = I,
	\end{aligned}
	\label{eq:opt:dr:gevo-max-3}
\end{equation}
where $A=L\inv QL\invtrnsp\in\psdsetm$. The solution to \eqref{eq:opt:dr:gevo-max-3} is given by the eigenvectors corresponding to the largest eigenvalues of $A$ \cite{Anstreicher2000SIAM}. By transforming back using $\Phi=\M L$, the solution to the original problem in \eqref{eq:opt:dr:gevo-max-2} is obtained. The solution includes a \abbrGEVP and is summarized in the following theorem.

\begin{thm} \label{thm:gevo:solution}
Assume $Q\in\psdsetnb$ and $S\in\pdsetnb$. Let $\M\in\realsmnb$, where $m\leq \nb$ and $\rank(\M)=m$. The solution to
\begin{equation}
	\begin{aligned}
		& \underset{\M}{\maximize} & & \trace\left( (\M S\Mt)\inv \M\trnsp Q\Mt \right),
	\end{aligned}
	\label{eq:thm:gevo:solution}
\end{equation}	
is given by $\Mopt=\col(u_1\trnsp,\dots,u_m\trnsp)$, where $u_i$ is a generalized eigenvector associated with $\lambda_i(Q,S)$, and $\lambda_1\geq\dots\geq\lambda_{\nb}$.
\end{thm}

% CHANGE OF BASIS
\subsubsection{Change of Basis}

For efficient communication it is important that $\M$ has orthogonal rows and that $\RM$ is a diagonal matrix. This has to do with the encoding of $(\yM,\RM)$ and $\M$ upon transmission, and is further explained in Sec.~\ref{sec:communication}. As shown in the next it is always possible to transform $\M$ such that these criteria are satisfied. We start with the following proposition.

\begin{prop} \label{prop:change-of-basis}
Let $S\in\pdsetn$ and $A,B\in\realsmn$, where $m\leq n$. If $\rowspan(A)=\rowspan(B)$, then
\begin{align}
	A\trnsp(ASA\trnsp)\inv A &= B\trnsp(BSB\trnsp)\inv B.
\end{align}
\end{prop}

\begin{proof}
Since $\rowspan(A)=\rowspan(B)$ there exists an invertible matrix $T$ such that $B=TA$. Hence
\begin{align*}
	A\trnsp(A SA\trnsp)\inv A 
	&= A\trnsp T\trnsp T\invtrnsp (ASA\trnsp)\inv T\inv T A \\
	&= A\trnsp T\trnsp (TASA\trnsp T\trnsp)\inv TA \\
	&= B\trnsp(BSB\trnsp)\inv B.
\end{align*}
\end{proof}

Consider $\Phi=\col(u_1\trnsp,\dots,u_m\trnsp)$, where $u_1,\dots,u_m$ are derived as in Theorem~\ref{thm:gevo:solution}. The rows of $\Phi\in\realsmnb$ span an $m$-dimensional subspace $\calV=\rowspan(\Mopt)\subseteq\realsnb$. In this sense, $\Mopt$ is a transformation
\begin{equation}
	\Phi \colon \realsnb \rightarrow \calV.
\end{equation}
For any two solutions $u_i$ and $u_j$ to $Qu = \lambda Su$, it is true that $u_i\trnsp Su_j=0$, while in general $u_i\trnsp u_j\neq0$ for $i\neq j$ \cite{Parlett1998}. Hence, $\Phi$ is not an orthogonal basis for $\calV$. Proposition~\ref{prop:change-of-basis} ensures that if $T$ is invertible, then $\Moptt(\Mopt R_2\Moptt)\inv\Mopt$ does not change if $\Mopt$ is substituted for $T\Mopt$. Moreover, since $\rowspan(\Mopt)=\rowspan(T\Mopt)$ \cite{Golub2013}, this can be exploited to derive $\Mopt$ such that $\Mopt R_2\Moptt$ is diagonal.
% \footnote{Ensuring that $\Mopt R_2\Moptt$ is diagonal is important from a communication perspective as described in Sec.~\ref{sec:communication}.}

A \emph{Gram-Schmidt procedure} can be expressed as $\Omega\trnsp=\Phi\trnsp T$, where $\Omega$ has orthonormal rows and $T$ is invertible \cite{Golub2013}. Thereby, Proposition~\ref{prop:change-of-basis} applies. Let $\Omega R_2\Omega=U\Sigma U\trnsp$ be an eigendecomposition, where $\Sigma\in\pdsetm$ is diagonal. By construction, $U\trnsp \Omega R_2 \Omega\trnsp U$ is diagonal. Hence, with $\Mopt=U\trnsp \Omega$, the covariance $\RM=\Mopt R_2\Moptt$ is diagonal. If $\Mopt=U\trnsp \Omega$, then since $U\trnsp$ is orthogonal and $\Omega$ has orthonormal rows
\begin{equation*}
	\Mopt\Moptt = U\trnsp \Omega\Omega\trnsp U = U\trnsp IU = I,
\end{equation*}
\ie, $\Mopt$ has orthonormal rows. 

%The change of basis procedure is illustrated in Fig.~\ref{fig:gevo:change-of-basis}, where $\RM=\Mopt R_2\Moptt$, $R_\Phi=\Phi R_2\Phi\trnsp$, and $R_\Omega=\Omega R_2\Omega\trnsp$.
%
%\begin{figure}[tb]
%	\centering
%	\begin{tikzpicture}[scale=0.45]
%		\input{fig/example_change_of_basis.tex}
%	\end{tikzpicture}
%	\caption{Change of basis. The information $\Mt(\M R_2\Mt)\inv \M$ is invariant to change of basis. The quantities $\Phi\trnsp R_\Phi\inv \Phi\trnsp$, $\Omega\trnsp R_\Omega\inv \Omega\trnsp$, and $\Moptt\RM\inv\Mopt$ are projected onto $\calV$ using $\Mopt$.}
%	\label{fig:gevo:change-of-basis}
%\end{figure}

% PROPOSED FRAMEWORK
\subsubsection{Proposed Framework}

The proposed \emph{generalized eigenvalue optimization} (\abbrGEVO) methodology for deriving $\Mopt$ is given in Algorithm~\ref{alg:gevo}. It is a direct application of Theorem~\ref{thm:gevo:solution} followed by the previously discussed change of basis procedure. Step~4 ensures that $\Mopt R_2\Moptt$ is diagonal. 

\begin{algorithm}[tb]
	\caption{\abbrGEVO}
	\label{alg:gevo}
	\begin{footnotesize}
	\begin{algorithmic}[0]
		\Input $R_1\in\pdsetnx$, $R_2\in\pdsetnb$, $R_{12}\in\reals^{\nx\times\nb}$, $H_2\in\reals^{\nb\times\nx}$, and $m$ 
		\begin{enumerate}
			\item Let $Q = (R_1H_2\trnsp-R_{12})\trnsp(R_1H_2\trnsp-R_{12})$ and $S = H_2R_1H_2\trnsp+R_2-H_2R_{12}-R_{21}H_2\trnsp$. 
			\item Compute $\lambda_1,\dots,\lambda_{\nb}$ and $u_1,\dots,u_{\nb}$ by solving $Qu = \lambda Su$.
			\item Define $\Phi=\col(u_1\trnsp,\dots,u_m\trnsp)$, and compute $\Omega=\col(v_1\trnsp,\dots,v_m\trnsp)$ such that $v_i\trnsp v_j=\delta_{ij}$ and $\rowspan(\Omega)=\rowspan(\Phi)$.
			\item Compute $\Omega R_2\Omega\trnsp = U\Sigma U\trnsp$ and let $\Mopt=U\trnsp\Omega$.
		\end{enumerate}
		\Output $\Mopt$
	\end{algorithmic}
	\end{footnotesize}
\end{algorithm}

%As an example of \abbrGEVO, assume that $\nx=\nb=2$, $H_2=I$, $m=1$, and $R_{12}=0$. Let $R_1$ and $R_2$ be defined according to their ellipses in Fig.~\ref{fig:ex:dr:gevo-example}. Solving $Qu=\lambda Su$, where $Q=R_1^2$ and $S=R_1+R_2$, yields two solutions $\lambdamin$ and $\lambdamax$ with associated generalized eigenvectors $\umin$ and $\umax$, respectively. The quantity of interest is $\calV=\rowspan(\umax\trnsp)$.
%
%\begin{figure}[tb]
%	\centering
%	\begin{tikzpicture}[scale=.6]
%		\input{fig/example_gevo_method.tex}
%	\end{tikzpicture}
%	\caption{Example of the \abbrGEVO method. Since $m=1$, $\Mopt=\umax\trnsp$. The vector $\umax$ spans the one-dimensional subspace $\calV$. Principal components of $R_1$ and $R_2$ are illustrated by dashed lines.}
%	\label{fig:ex:dr:gevo-example}
%\end{figure}	

\begin{rmk}
Optimality guarantees are developed under the stated model assumptions. If the model assumptions are incorrect, then it is generally impossible to make any definite claims. However, from matrix perturbation theory, we have that the eigenvectors of a slightly, in a relative sense, perturbed eigenvalue problem are approximately equal to the unperturbed case \cite{Stewart1990}.
\end{rmk}

\begin{rmk}
The derivation of the \abbrGEVO method demonstrates that, by transforming the original problem, an \abbrEVP can equivalently be solved for the derivation of $\Mopt$. However, the \abbrGEVP formulation is kept for numerical considerations \cite{Parlett1998}. %The proposed \abbrGEVP can, \eg, be solved by the \emph{QZ algorithm} \cite{Moler1973SIAM}.
\end{rmk}

% SINGULAR CASE
\subsubsection{Singular Case}

If $S\succeq0$ is singular, then $\M S\Mt$ might be singular. If so, the inverse in \eqref{eq:gevo:P-to-minimize-2} is replaced by a pseudoinverse such that
\begin{equation}
	P = R_1 - \Delta\Mt(\M S\Mt)\pinv\M\Delta\trnsp,
	\label{eq:gevo:singular:P-to-minimize}
\end{equation}
where $\Delta=R_1H_2\trnsp-R_{12}$. Assume that $1\leq m\leq \rank(S)=r$. Let $S=VDV\trnsp$, where $D=\diag(D_1,0)$, $D_1\in\pdset^{r}$, and $0$ is an $(\nb-r)\times(\nb-r)$ matrix of zeros, such that
\begin{equation}
	S\pinv = VD\pinv V\trnsp = V\diag(D_1\inv,0)V\trnsp.
\end{equation}
The case $m>r$ does not make sense since this does not improve $P$ in \eqref{eq:gevo:singular:P-to-minimize} for the same reason that $m>\nb$ does not make sense if $S\in\pdsetnb$. This is stated formally in Proposition~\ref{prop:gevo:singular-case}. 

\begin{prop} \label{prop:gevo:singular-case}
Assume $\rank(S)=\rank(\M_1)=r$, where $S\in\psdsetnb$ and $\M_1\in\reals^{r\times\nb}$. Let $\M_1S\Mt_1\in\pdset^r$. If $\M=\col(\M_1,\M_2)$, where $\M_2\in\reals^{r\times\nb}$, then
\begin{equation}
	\Mt(\M S\Mt)\M = \Mt_1(\M_1 S\Mt_1)\M_1.
\end{equation}
\end{prop}

\begin{proof}
Let $\M_2=\col(\M_a,\M_b)$, where $\M_a$ and $\M_b$ lie in the column space and in the null space of $S$, respectively. Hence, $\M_a=A\M_1$ and $\M_b S=0$. Assume \wolog that $\M_1S\Mt_1=I$. In this case, $\M S\Mt=BB\trnsp$ 
%\begin{equation*}
%	\M S\Mt = \BBM I & A\trnsp & 0 \\ A & AA\trnsp & 0 \\ 0 & 0 & 0 \EBM = BB\trnsp,
%\end{equation*}
where $B=\BBM I&A\trnsp&0\EBM\trnsp$. By construction, $B$ is full column rank, which implies that $B\pinv=(B\trnsp B)\inv B\trnsp$ with $B\trnsp B=I+A\trnsp A$. Moreover, $(BB\trnsp)\pinv = (B\trnsp)\pinv B\pinv$ \cite{Greville1966SIAM}. Hence
\begin{align*}
	&\Mt(\M S\Mt)\pinv\M \\ 
	&\quad= \Mt (BB\trnsp)\pinv \M = \Mt B(B\trnsp B)\inv(B\trnsp B)\inv B\trnsp \M \\ % \Mt (B\trnsp)\pinv B\pinv \M = 
	&\quad= \BBM \M_1\\ A\M_1A \\ \M_b\EBM\trnsp \BBM I\\A\\0\EBM (I+A\trnsp A)^{-2} \BBM I\\ A \\0\EBM\trnsp \BBM \M_1\\A\M_1\\\M_b\EBM \\
	&\quad= \Mt_1 (I+A\trnsp A)(I+A\trnsp A)^{-2}(I+A\trnsp A)\M_1 \\
	&\quad= \Mt_1\M_1 = \Mt_1(\M_1 S\Mt_1)\inv\M_1.
\end{align*}
\end{proof}

Let $\Phi\trnsp = D\supsqrt V\trnsp \Mt$. If $m\leq r$ and $\rank(\M)=m$, then it is still possible to impose the constraint $\M S\Mt=I$ such that $\Phi\Phi\trnsp=I$. Hence, Algorithm~\ref{alg:gevo} is valid even if $S$ is singular, given that $m\leq r$ and $\rank(\M)=m$.

\begin{rmk}
If $m\leq r$, then it might still be that $\M Q\Mt=0$ for all feasible $\M$ since $Q$ might be singular. In this degenerated case, fusion of $(y_1,R_1)$ with $(\yM,\RM)$ yields no improvement compared to using $(y_1,R_1)$ alone.
\end{rmk}

% CHOOSING m
\subsubsection{Choosing $m$}

The parameter $m$ is essentially a design choice. As stated previously, there is no benefit of using $m>\rank(S)=r$, where $r\leq\nb$. In some cases, it might be desirable to choose $m$ adaptively. By construction
\begin{equation}
	\trace\left( \Delta\Mt(\M S\Mt)\inv\M\Delta\trnsp \right) = \trace\left( (\M S\Mt)\inv \M Q\Mt \right) = \sum_{i=1}^m \lambda_i,
\end{equation}
where $\M=\Mopt=\col(u_1\trnsp,\dots,u_m\trnsp)$, $(\lambda_i,u_i)$ is given by $Qu=\lambda Su$, and $\lambda_1\geq\dots\geq\lambda_{\nb}$. Consider now $\M=\Mopt$ as a function of $m$. By defining $\ell_0=\trace(R_1)$ and
\begin{equation}
	\ell_m = \trace(P) = \trace(R_1) - \trace\left( (\M S\Mt)\inv \M Q\Mt \right) = \ell_0 - \sum_{i=1}^m\lambda_i,
\end{equation}
for $m\in\{1,2,\dots,r\}$, it is possible to relate $m$ directly to the fusion gain with $\ell_m$. For instance, $m$ can be chosen adaptively to be the smallest integer $m$ such that
\begin{equation}
	\frac{\ell_0 - \ell_{m}}{\ell_0 - \ell_{r}} = \frac{\sum_{i=1}^m\lambda_i}{\sum_{i=1}^r\lambda_i} \geq \tau, 
\end{equation}
for some threshold $\tau\in[0,1]$.

% COMP COMPLEXITY
\subsubsection{Computational Complexity} \label{sec:gevo:computational-complexity}

The \abbrPCA in Algorithm~\ref{alg:pco} can be solved using the \emph{QR algorithm}, which has a computational complexity of $\calO(\nb^3)$ \cite{Trefethen1997NLA}. The \abbrGEVP in Algorithm~\ref{alg:gevo} is conveniently solved using the \emph{QZ algorithm}, which is also $\calO(\nb^3)$ \cite{Moler1973SIAM}. However, \abbrGEVO involves a number of additional steps related to the change of basis. The cost of the Gram-Schmidt procedure is $\calO(\nb^2m)$ \cite{Golub2013} and the cost of the subsequent eigendecomposition is $\calO(m^3)$. 

% GEVO FOR THE CONSIDERED FUSION METHODS

% GEVO-KF
\subsection{\abbrGEVO for Kalman Fusion} \label{sec:gevo-kf}

The \abbrGEVO method for \abbrKF is automatically retrieved by setting $R_{12}=0$ in Algorithm~\ref{alg:gevo}. This yields
\begin{align*}
	Q &= H_2R_1^2H_2\trnsp, & S &= H_2R_1H_2\trnsp+R_2. 
\end{align*}
\GEVOKF, \ie, \abbrGEVO in the \abbrKF case, is provided in Algorithm~\ref{alg:gevo-kf}. \abbrKF assumes zero cross-correlations and is therefore applicable when some kind of decorrelation procedure is utilized. For instance, by using the methods in \cite{Tian2010Fusion} or \cite{Forsling2023Aero}.

\begin{algorithm}[t]
	\caption{\GEVOKF} 
	\label{alg:gevo-kf}
	\begin{footnotesize}
	\begin{algorithmic}[0]
		\Input $R_1\in\pdsetnx$, $R_2\in\pdsetnb$, $R_{12}\in\reals^{\nx\times\nb}$, $H_2\in\reals^{\nb\times\nx}$, and $m$ 
		\begin{enumerate}
			\item Let $Q = H_2R_1^2H_2\trnsp$ and $S = H_2R_1H_2\trnsp+R_2$. 
			\item Compute $\lambda_1,\dots,\lambda_{\nb}$ and $u_1,\dots,u_{\nb}$ by solving $Qu = \lambda Su$.
			\item Define $\Phi=\col(u_1\trnsp,\dots,u_m\trnsp)$, and compute $\Omega=\col(v_1\trnsp,\dots,v_m\trnsp)$ such that $v_i\trnsp v_j=\delta_{ij}$ and $\rowspan(\Omega)=\rowspan(\Phi)$.
			\item Compute $\Omega R_2\Omega\trnsp = U\Sigma U\trnsp$ and let $\Mopt=U\trnsp\Omega$.
		\end{enumerate}
		\Output $\Mopt$
	\end{algorithmic}
	\end{footnotesize}
\end{algorithm}

% GEVO-CI
\subsection{\abbrGEVO for Covariance Intersection} \label{sec:gevo-ci}

For fusion using \abbrCI, the \abbrGEVO method does not apply directly. The reason is the dependency on $\omega$. In \cite{Forsling2022Fusion}, an iterative algorithm based on \emph{alternating minimization} (\abbrAM, \cite{Beck2017}) is proposed. First, it is noted that \abbrCI only differs from \abbrKF through the $\omega$ parameter. The basic idea is to alternate between keeping $\omega$ and $\Phi$ fixed. A generalization for $m\geq1$ of the original algorithm proposed in \cite{Forsling2022Fusion} is provided in Algorithm~\ref{alg:gevo-ci}. The loss function value $J_k$ is defined in \eqref{eq:gevo-ci:J}.

\begin{algorithm}[t]
	\caption{\GEVOCI}
	\label{alg:gevo-ci}
	\begin{footnotesize}
	\begin{algorithmic}[0]
		\Input $\omega_0$, $R_1\in\pdsetnx$, $R_2\in\pdsetnb$, $H_2\in\realsnbnx$, $m$, $k=0$, and $\epsilon$ \\
		\begin{enumerate}
			\item Let $k\leftarrow k+1$. Compute $\lambda_1,\dots,\lambda_{\nb}$ and $u_1,\dots,u_{\nb}$ by solving $Qu = \lambda Su$, where $Q = H_2R_1^2H_2\trnsp/\omega_{k-1}^2$ and $S = H_2R_1H_2\trnsp/\omega_{k-1}+R_2/(1-\omega_{k-1})$. Let $\Phi_k=\col(u_1\trnsp,\dots,u_{m}\trnsp)$, where $u_i$ is a generalized eigenvector associated with $\lambda_i$. 
			\item Let $R_\Phi=\Phi_k R_2 \Phi_k\trnsp$. Compute $\omega_k$ by solving
			\begin{equation*}
				\underset{\omega}{\minimize}\quad \trace\left(\left(\omega R_1\inv + (1-\omega)H_2\trnsp \Phi_k\trnsp R_\Phi\inv \Phi_k H_2\right)\inv\right).
			\end{equation*}
			\item Let $J_k$ be according to \eqref{eq:gevo-ci:Jk}. If $(J_{k-1}-J_k)/J_k>\epsilon$, then go back to step~1. Otherwise continue to step~4.
			\item Define $\Phi=\col(u_1\trnsp,\dots,u_m\trnsp)$, and compute $\Omega=\col(v_1\trnsp,\dots,v_m\trnsp)$ such that $v_i\trnsp v_j=\delta_{ij}$ and $\rowspan(\Omega)=\rowspan(\Phi)$.
			\item Compute $\Omega R_2\Omega\trnsp = U\Sigma U\trnsp$ and let $\M=U\trnsp\Omega$.
		\end{enumerate}
		\Output $\M$
	\end{algorithmic}
	\end{footnotesize}
\end{algorithm}

The parameter $\epsilon$ in Algorithm~\ref{alg:gevo-ci} is a design parameter chosen as a compromise between computational speed and exactness of the solution given by the final iterate. Here, $\omega_0=1/2$ is used.

% CONVERGENCE ANALYSIS
\subsubsection{Convergence Analysis}

The \abbrAM method in Algorithm~\ref{alg:gevo-ci} alternates between solving two different kinds of optimization problems. Each separate problem is well-posed, and a global minimum is obtained. The solution to the problem in step~1 is given by the generalized eigenvalues associated with the largest generalized eigenvalues of a \abbrGEVP. Step~2 involves solving a convex optimization problem \cite{Orguner2017SDF} for which any local minimum is also a global minimum. This does, however, not imply that the final iteration of Algorithm~\ref{alg:gevo-ci} is a global minimizer. Below, it is shown that the iterations of Algorithm~\ref{alg:gevo-ci} converge to a stationary point.

Let $J(\omega,\Phi)=\trace(P)$ with $P$ according to \eqref{eq:method:ci-dr}, \ie,
\begin{small}
\begin{align}
	J(\omega,\Phi) 
	&= \trace\left(\left(\omega R_1\inv + (1-\omega)H_2\trnsp\Phi\trnsp(\Phi R_2\Phi\trnsp)\inv\Phi H_2 \right)\inv\right). 
	\label{eq:gevo-ci:J}
\end{align}	
\end{small}
Define
\begin{align}
	J_{k-\frac{1}{2}} &= J(\omega_{k-1},\Phi_k), & J_{k} &= J(\omega_k,\Phi_k). \label{eq:gevo-ci:Jk}
\end{align}
Consider a sequence of $\nk$ iterations and hence $2\nk$ subiterations. Each iteration and subiteration have the same feasible set. Since in each subiteration a minimum is obtained it is concluded that
\begin{equation*}
	J_{\frac{1}{2}} \geq J_1 \geq \cdots \geq J_{N-\frac{1}{2}} \geq J_{N},
\end{equation*}
which is a nonincreasing sequence $\{J_k\}$. Denote by $\{(\omega,\Phi)_k\}$ the sequence of points generating $\{J_k\}$. Since $P\in\pdsetnx$, there exists a lower bound $\Jlow>0$ such that $J(\omega,\Phi)\geq\Jlow,\forall(\omega,\Phi)$. Hence, the \emph{monotonic convergence theorem} \cite{Bibby1974GMJ} is applicable, which states that 
\begin{equation}
	\lim_{k\rightarrow\infty} J_k = J(\omegalim,\Philim),
\end{equation}
where $(\omegalim,\Philim)$ is a limit point. In the limit $J_{k+\frac{1}{2}}-J_{k-\frac{1}{2}}\rightarrow0$ and $J_{k+1}-J_{k}\rightarrow0$. Since $J(\omega,\Phi)$ is differentiable \wrt $\omega$ and $\Phi$ on its domain, this implies that 
\begin{align}
	\left.\frac{\partial}{\partial\omega} J(\omega,\Phi) \right|_{\omega=\omegalim} &= 0, & \left.\frac{\partial}{\partial\Phi} J(\omega,\Phi) \right|_{\Phi=\Philim} &= 0,
\end{align}
and hence $(\omegalim,\Philim)$ is a stationary point. The convergence results are summarized in Theorem~\ref{thm:gevo-ci:convergence}. It should be emphasized that from Theorem~\ref{thm:gevo-ci:convergence} alone, it cannot be concluded if $(\omegalim,\Philim)$ is a local minimizer, global minimizer or a saddle point.

% THM: GEVO CI CONVERGENCE
\begin{thm} \label{thm:gevo-ci:convergence}
Let $\{(\omega,\Phi)_k\}$ be a sequence of points generated by Algorithm~\ref{alg:gevo-ci} and let $J(\omega,\Phi)$ be given by \eqref{eq:gevo-ci:J}. Then $\{(\omega,\Phi)_k\}$ converges to a stationary point $(\omegalim,\Philim)$ of $J(\omega,\Phi)$.
\end{thm}

The convergence rate is evaluated numerically. In each simulation, $R_1,R_2\in\pdsetnx$ are \iid according to $R_1,R_2\sim \calW(I,\nx)$, where $\calW(I,\nx)$ is the \emph{Wishart distribution} of $\nx$ degrees of freedoms \cite{Wishart1928Biometrika}. If in a certain sample $R_2\succeq R_1$, then $R_1$ is resampled until $R_2\not\succeq R_1$. This resampling is done because $R_2\succeq R_1$ would trivially yield $\omega=1$ in Algorithm~\ref{alg:gevo-ci} for all feasible $\M$.

The evaluation is performed with $\nx\in\{6,9\}$ and $\epsilon\in\{0.1\%,0.01\%\}$. In each case, and for each $m$, 1\,000\,000 \abbrMC simulations are conducted. The evaluation measure is the number of iterations until $|J_{k-1}-J_k|/J_k\leq\epsilon$. The results are shown in Fig.~\ref{fig:dr:gevo-ci:cra}. The statistics are summarized in Table~\ref{tab:gevoci:convergence-analysis}, where \emph{typ}, \emph{mean}, and \emph{std} refer to the typical value, the mean, and the standard deviation, respectively. The \GEVOCI algorithm converges very fast in general, \eg, for $\nx=9$ and $\epsilon=0.01\%$ with $m=3$ the mean number of iterations is approximately $4.056\pm0.634$ and the typical value is 4. The convergence rate is improved by increasing $m$. 

The number of iterations increases as $\nx$ increases and as $\epsilon$ decreases. Hence, it cannot be expected that the computational cost of \GEVOCI scales as $\calO(\nb^3)$. In addition, \GEVOCI involves a convex optimization problem, which further increases the computational cost. This can however be improved using, \eg, the approximation proposed in \cite{Orguner2017SDF}.

\begin{figure}[tb]
	\centering
	\begin{tikzpicture}[xscale=.16,yscale=.5]
		\input{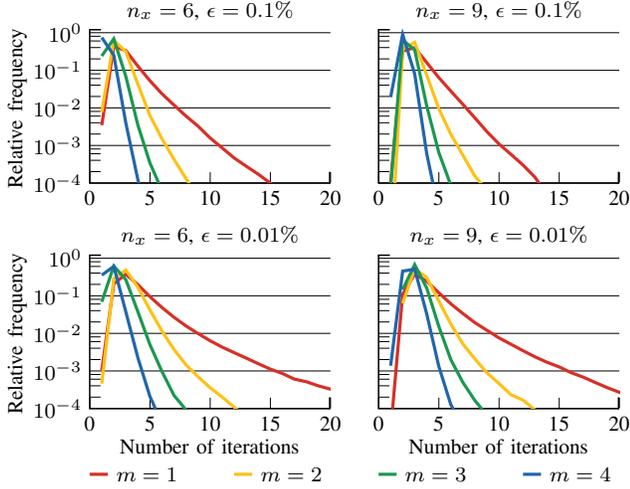}
	\end{tikzpicture}
	\caption{Results of the numerical convergence rate analysis, where $\nx\in\{6,9\}$ and $\epsilon\in\{0.1\%,0.01\%\}$. The plots show the relative frequency of the number of iterations until the criterion $(J_{k-1}-J_k)/J_k\leq\epsilon$ in Algorithm~\ref{alg:gevo-ci} is met.}
	\label{fig:dr:gevo-ci:cra}
\end{figure}

\begin{table}[htb] 
\caption{Convergence Rate Analysis} %
\centering
\label{tab:gevoci:convergence-analysis}
\begin{footnotesize}
\begin{minipage}{.49\columnwidth}
\subcaption*{$\nx=6$, $\epsilon=0.1\%$}
\begin{tabular}{cccc}
\toprule\midrule
$m$ & \textbf{typical} & \textbf{mean} & \textbf{std} \\
\midrule
1 & 3 & 3.995 & 1.311 \\ 
2 & 3 & 3.421 & 0.645 \\ 
3 & 3 & 2.836 & 0.549 \\ 
4 & 2 & 2.260 & 0.449 \\ 
\bottomrule
\end{tabular}
\vspace{0.5em}
\subcaption*{$\nx=6$, $\epsilon=0.01\%$}
\vspace{0.5em}
\begin{tabular}{cccc}
\toprule\midrule
$m$ & \textbf{typical} & \textbf{mean} & \textbf{std} \\
\midrule
1 & 4 & 4.795 & 2.099 \\ 
2 & 4 & 4.035 & 0.961 \\ 
3 & 3 & 3.298 & 0.691 \\ 
4 & 3 & 2.687 & 0.554 \\ 
\bottomrule
\end{tabular}
\end{minipage} %
\begin{minipage}{.49\columnwidth}
\subcaption*{$\nx=9$, $\epsilon=0.1\%$}
\begin{tabular}{cccc}
\toprule\midrule
$m$ & \textbf{typical} & \textbf{mean} & \textbf{std} \\
\midrule
1 & 4 & 4.156 & 1.241 \\ 
2 & 4 & 3.744 & 0.663 \\ 
3 & 3 & 3.376 & 0.513 \\ 
4 & 3 & 3.062 & 0.313 \\ 
\bottomrule
\end{tabular}
\vspace{0.5em}
\subcaption*{$\nx=9$, $\epsilon=0.01\%$}
\vspace{0.5em}
\begin{tabular}{cccc}
\toprule\midrule
$m$ & \textbf{typical} & \textbf{mean} & \textbf{std} \\
\midrule
1 & 4 & 5.108 & 2.059 \\ 
2 & 4 & 4.487 & 0.914 \\ 
3 & 4 & 4.056 & 0.634 \\ 
4 & 4 & 3.582 & 0.568 \\  
\bottomrule
\end{tabular}
\end{minipage}
\end{footnotesize}
\end{table}

% GEVO-LE
\subsection{\abbrGEVO for the Largest Ellipsoid Method} \label{sec:gevo-le}

\begin{algorithm}[t]
	\caption{\GEVOLE}
	\label{alg:gevo-le}
	\begin{footnotesize}
	\begin{algorithmic}[0]
		\Input $R_1\in\pdsetnx$, $R_2\in\pdsetnb$, $H_2\in\realsnbnx$, and $m$ \\
		\begin{enumerate}
			\item Transform into the information domain
			\begin{align*}
				\calI_1 &= R_1\inv, & \calI_2 &= H_2\trnsp R_2\inv H_2.
			\end{align*}
			\item Factorize $\calI_1=U_1\Sigma_1 U_1\trnsp$ and let $T_1=\Sigma_1^{-\frac{1}{2}} U_1\trnsp$. Factorize $T_1\calI_2T_1\trnsp=U_2\Sigma_2U_2\trnsp$ and let $T_2=U_2\trnsp$. Transform using $T=T_2T_1$ according to
			\begin{align*}
				\calI_1' &= T\calI_1T\trnsp = I, & \calI_2' &= T\calI_2T\trnsp.
			\end{align*}
			\item Let $D$ be diagonal. For each $i=1,\dots,\nx$ compute
			\begin{equation*}
				[D]_{ii} = \min\left( 1, [\calI_2']_{ii}  \right).
			\end{equation*}
			\item Let $\calI_\gamma = T\inv DT\invtrnsp$ and $R_{12} = R_1\calI_\gamma H_2\trnsp R_2$. Compute $\M$ using Algorithm~\ref{alg:gevo} with inputs $R_1$, $R_2$, $R_{12}$, $H_2$, and $m$.
		\end{enumerate}
		\Output $\M$
	\end{algorithmic}
	\end{footnotesize}
\end{algorithm}

To be able to use the \abbrGEVO framework with the \abbrLE method, some adaptations are required. As discussed in \cite{Noack2016MFI}, the \abbrLE method makes an implicit assumption about $R_{12}$ if $H_2=I$. This concept can be generalized by assuming the \emph{common information decomposition}\footnote{This decomposition is a generalization of the one utilized in, \eg, \cite{Noack2016MFI}.}
\begin{subequations} \label{eq:common-information-decomposition}
\begin{align}
  	R_1\inv &= \calI_1^e + \calI_\gamma, &
	R_1\inv y_1 &= \calI_1^e y_1^e + \calI_\gamma\gammahat, \\
	R_2\inv &= \calI_2^e + H_2\calI_\gamma H_2\trnsp, &
	R_2\inv y_2 &= \calI_2^e y_2^e + H_2\calI_\gamma\gammahat.
\end{align}	
\end{subequations}
where $\cov(y_1^e,y_2^e)=\cov(y_1^e,\gammahat)=0$ and $\cov(y_2^e,\gammahat)=0$. Define $=R_i^e=\cov(\ytilde_i^e)$ and $\Gamma=\cov(\gammatilde)$. Allow for the \emph{exclusive information} $\calI_i^e$ and the \emph{common information} $\calI_\gamma$ to be singular, but still assume $R_1\in\pdsetnx$ and $R_2\in\pdsetnb$. Assume that $\calI_\gamma$ has $p\leq\nx$ zero eigenvalues, \ie, $\Gamma=\cov(\gammahat)$ is infinite in $p$ components\footnote{Infinite variance corresponds to zero information.}. Hence, the issue of infinite variance is conveniently handled in the information domain. The information $\calI_\gamma$ is given by
\begin{align}
	\calI_\gamma &= VDV\trnsp, & D &=
	\diag(d_1,\dots,d_{\nx-p},0,\dots,0)\in\psdsetnx,
\end{align}
where $d_i>0$ and $V=\BBM v_1&\dots&v_{\nx}\EBM$ is an orthogonal matrix. In this configuration, the components with infinite variance correspond to the directions $v_{\nx-p+1},\dots,v_{\nx}$. Let $\Phi=\col(v_1\trnsp,\dots,v_{\nx-p}\trnsp)$ such that $\calI_\gamma=\Phi\trnsp(\Phi\Gamma\Phi\trnsp)\inv\Phi$, where $\Phi\Gamma\Phi\trnsp$ is finite and invertible by construction. It follows that
\begin{align}
	\calI_\gamma\Gamma\calI_\gamma 
	&= \Phi\trnsp(\Phi\Gamma\Phi\trnsp)\inv\Phi\Gamma \Phi\trnsp(\Phi\Gamma\Phi\trnsp)\inv\Phi
	= \Phi\trnsp(\Phi\Gamma\Phi\trnsp)\inv\Phi \nonumber \\
	&= \calI_\gamma.
	\label{eq:limit:IGI}
\end{align}
Let $\ytilde_i=y_i-\EV(y_i)$, $\ytilde_i^e=y_i^e-\EV(y_i^e)$, and $\gammatilde=\gammahat-\EV(\gammahat)$. Using \eqref{eq:common-information-decomposition} and \eqref{eq:limit:IGI} yield
\begin{align}
	R_{12} \nonumber 
	&= \EV\left(\left(R_1\left(\calI_1^e y_1^e + \calI_\gamma\gammahat\right) - x\right) \vphantom{H\trnsp} \right. \\ 
	&\quad\times \left.\left(R_2\left(\calI_2^e y_2^e + H_2\calI_\gamma\gammahat\right) - H_2x\right)\trnsp \right) \nonumber \\
	&= R_1\EV\left( \left(\calI_1^e \ytilde_1^e + \calI_\gamma\gammatilde\right) \left(\calI_2^e \ytilde_2^e + H_2\calI_\gamma\gammatilde\right)\trnsp \right)R_2 \nonumber \\  
	&= R_1\calI_\gamma \EV(\gammatilde\gammatilde\trnsp)  \calI_\gamma H_2\trnsp R_2 = R_1\calI_\gamma\Gamma\calI_\gamma H_2\trnsp R_2 \nonumber \\ 
	&= R_1\calI_\gamma H_2\trnsp R_2. \label{eq:le:R12}
\end{align}

The \abbrGEVO method for \abbrLE is provided in Algorithm~\ref{alg:gevo-le}. The key step is to compute the common information $\calI_\gamma$ as above. This corresponds to steps~1--3 in Algorithm~\ref{alg:gevo-le}. Then $R_{12}$ is given by \eqref{eq:le:R12}. When $R_{12}$ has been obtained, Algorithm~\ref{alg:gevo} is applicable. If $\rank(H_2)=\nx$, then both $\calI_1$ and $\calI_2$ in Algorithm~\ref{alg:gevo-le} are full rank. Hence, also $\calI_\gamma\inv=\Gamma$ is full rank. 

In \cite{Forsling2022TSP} it is noted that if $H_2=I$, then \abbrLE corresponds to a particular $R_{12}$. As a consequence, $\M$ computed by Algorithm~\ref{alg:gevo} is guaranteed to be optimal given that \abbrLE is used to compute $(\xhat,P)$. That is, $\M$ computed by Algorithm~\ref{alg:gevo-le} is optimal \wrt $\trace(P)$, if $H_2=I$ and $P$ is computed using \eqref{eq:method:le-dr:estimate}. The \abbrLE method is an optimal conservative estimation when the following conditions hold: (i) $H_2=I$; and (ii) $R_{12}$ is according to the \emph{componentwise aligned correlations} structure defined in \cite{Forsling2022TSP}. This means that $\M$ computed by Algorithm~\ref{alg:gevo-le} implies an optimal conservatively fused estimate if (i) and (ii) hold simultaneously.

Algorithm~\ref{alg:gevo-le} involves several additional steps for computation of $R_{12}$, compared to Algorithm~\ref{alg:gevo}. However, since these operations are also EVPs, the computational cost of \GEVOLE scales as $\calO(\max(\nx^3,\nb^3))$.

% THEORETICAL COMPARISON
\subsection{Theoretical Comparison of \abbrGEVO and \abbrPCO}

The \abbrGEVO method in Algorithm~\ref{alg:gevo} computes $\Mopt$, which is optimal \wrt $\trace(P)$. In certain cases, $\Mopt$ is equal to $\Mpco$ computed by the \abbrPCO method in Algorithm~\ref{alg:pco}, but in general, $\Mpco\neq\Mopt$. The next example illustrates a case when $\Mpco=\Mopt$ and a case when $\Mpco$ is as far as possible from $\Mopt$.

Assume that $\nb=\nx=2$, $H_2=I$, and $R_{12}=0$. Let $R_2=\diag(4,1)$ and $R_1$ be defined either as: (i) $R_1=\diag(1,4)$; or (ii) $R_2=\diag(4,1)$. Let $\Mgevo$ be computed using Algorithm~\ref{alg:gevo-kf} and $\Mpco$ be computed using Algorithm~\ref{alg:pco}. Let $\umin$ and $\umax$ be the eigenvectors associated with the minimum and maximum eigenvalues $\lambdamin(R_2)$ and $\lambdamax(R_2)$, respectively. In case (i),
\begin{align*}
	\Mgevo &= \BBM0&1\EBM = \umin\trnsp, & \Mpco &= \BBM0&1\EBM = \umin\trnsp,
\end{align*}
such that $\Mpco=\Mgevo$. In case (ii)
\begin{align*}
	\Mgevo &= \BBM1&0\EBM = \umax\trnsp, & \Mpco &= \BBM0&1\EBM = \umin\trnsp,
\end{align*}
such that $\Mgevo\Mpco\trnsp=0$. That is, \abbrPCO yields $\Mpco=\umin\trnsp$, but to minimize $\trace(P)$ the $\umax\trnsp$ should be used. As demonstrated below, in case (ii), \abbrPCO yields the worst possible choice of $\M$ \wrt to $\trace(P)$.

The observation made in the previous example is a special case of a more general result. Assume that $H_2=I$ and $R_{12}=0$, such that $(y_1,R_1)$ and $(\yM,\RM)$ are fused optimally using \eqref{eq:method:kf-dr}. Let $R_1,R_2\in\pdsetnx$ be given by
\begin{subequations} \label{eq:dr:shared-ev}
\begin{align}
	R_1 &= V\Sigma V\trnsp, & 
	\Sigma &= \diag(\mu_1,\dots,\mu_{\nx}), \\
	R_2 &= V\Pi V\trnsp, & 
	\Pi &= \diag(\pi_1,\dots,\pi_{\nx}),
\end{align} 	
\end{subequations}
where $V=\BBM v_1&\dots&v_{\nx} \EBM$ is an orthogonal matrix, $\mu_1\geq\dots\geq\mu_{\nx}$ are the eigenvalues of $R_1$, and $\pi_1\geq\dots\geq\pi_{\nx}$ are the eigenvalues of $R_2$. That is, $R_1$ and $R_2$ share eigenvectors, and their eigenvalues are ordered correspondingly. In this case
\begin{align*}
	Q &= R_1^2 = V\Sigma^2V\trnsp, &
	S &= R_1+ R_2 = V(\Sigma+\Pi)V\trnsp.
\end{align*}
Since $S\in\pdsetnx$, the \abbrGEVP $Qu=\lambda Su$ solved in \abbrGEVO can be expressed as
\begin{equation}
	S\inv Q u = V(\Sigma+\Pi)\inv V\trnsp V \Sigma^2 V = V\Lambda V\trnsp u = \lambda u,
	\label{eq:dr:comparison:gevp-to-evp}
\end{equation}
where $\Lambda=\diag(\lambda_1,\dots,\lambda_{\nx})$ and $\lambda_i=\mu_i^2/(\mu_i+\pi_i)$. The assumed ordering of $\mu_i$ and $\pi_i$ implies that 
\begin{equation*}
	\frac{1}{\mu_1+\pi_1} \leq \dots \leq \frac{1}{\mu_{\nb}+\pi_{\nx}},
\end{equation*}
and hence $\lambda_1\geq\dots\geq\lambda_{\nx}$. Since $R_{12}=0$, Algorithm~\ref{alg:gevo-kf} is applicable, where the output $\Mopt$ minimizes $\trace(P)$ with $P$ given in \eqref{eq:method:kf-dr}. In this case $\Mopt = \col(v_1\trnsp,\dots,v_m\trnsp)$. Meanwhile, using Algorithm~\ref{alg:pco} yields $\Mpco = \col(v_{\nx-m+1}\trnsp,\dots,v_{\nx}\trnsp)$. However, by inspection of Algorithm~\ref{alg:gevo-kf} it can be seen that this $\Mpco$ in fact minimizes $\trace((\M S\Mt)\inv \M Q\Mt)$ and hence maximizes $\trace(P)$ among all feasible $\M\in\realsmnb$. Under the given assumptions, it is concluded that \abbrPCO provides the worst possible $\M$ if the goal is to minimize $\trace(P)$\footnote{The \abbrPCO method maximizes the \abbrMSE in this case.}. These results are summarized in Theorem~\ref{thm:pco-worst}.

\begin{thm} \label{thm:pco-worst}
Assume that $H_2=I$ and $R_{12}=0$, and that $R_1,R_2\in\pdsetnx$ are given according to \eqref{eq:dr:shared-ev}, where $\mu_1\geq\dots\geq\mu_{\nx}$ and $\pi_1\geq\dots\geq\pi_{\nx}$. Let $\M\in\realsmnx$ and $P$ be given according to \eqref{eq:method:kf-dr}. Then, $\Mpco$ computed by the \abbrPCO method in Algorithm~\ref{alg:pco} solves
\begin{equation}
	\begin{aligned}
		& \underset{\M}{\maximize} & & \trace(P) \\
		& \subjectto & & \M\M\trnsp = I.
	\end{aligned}	
\end{equation}
\end{thm}

In practice, it is not likely to have exactly the conditions assumed in Theorem~\ref{thm:pco-worst}. One relevant question is what can be expected when the conditions hold approximately. Consider $A$ and $B=A+\Delta$. From matrix perturbation theory it is known that if the eigenvalues of $A$ are distinct, then for small $\Delta$ such that $A\approx B$, the eigenvalues and eigenvectors $A$ and $B$ are approximately equal \cite{Stewart1990}. This can be utilized as follows. Assume
\begin{align*}
	R_1 &= V\Sigma V\trnsp + \Delta, &
	R_2 &= V\Pi V\trnsp - \Delta,
\end{align*}
where $V$, $\Sigma$, and $\Pi$ are defined as in \eqref{eq:dr:shared-ev}. Then
\begin{align*}
	Q &= V\Sigma V\trnsp + \Delta V\Sigma V\trnsp + V\Sigma V\trnsp \Delta + \Delta^2, &
	S &= V(\Sigma + \Pi)V\trnsp. 
\end{align*}
The matrix $\Lambda$ of the \abbrEVP defined in \eqref{eq:dr:comparison:gevp-to-evp} is now given by
\begin{align*}
	\Lambda_0 + V(\Sigma+\Pi)\inv V\trnsp\Delta V\Sigma V\trnsp + V(\Sigma+\Pi)\inv\Sigma V\trnsp\Delta + \Delta^2,
\end{align*}
where $\Lambda_0=V(\Sigma+\Pi)\inv\Sigma^2V\trnsp$. Assume that $\Delta$ is such that $\Lambda\approx\Lambda_0$ and $R_2\approx V\Pi V\trnsp$. If, in addition, the eigenvalues of each of $\Lambda_0$ and $V\Pi V\trnsp$ are distinct, then the following apply: (i) \abbrPCO with $R_2$ yields approximately the same solution as \abbrPCO with $V\Sigma V\trnsp$; and (ii) \abbrGEVO with $\Lambda$ yields approximately the same solution as \abbrGEVO with $\Lambda_0$. If so, Theorem~\ref{thm:pco-worst} is expected to be approximately true.

\begin{rmk}
In the previous reasoning, it is not stated what a small $\Delta$ quantitatively means. However, the point here is in a qualitative sense. That is, if \eqref{eq:dr:shared-ev} holds approximately, then \abbrPCO likely is a bad choice for \abbrDR. A \naive application of \abbrPCO hence might imply poor performance.	
\end{rmk}

%In case of \GEVOCI, a \abbrGEVP and a convex optimization problem are solved iteratively which of course increases the computational complexity further.

% --- COMMUNICATION ---
\section{Communication Considerations} \label{sec:communication}

For Agent~1 to use the received \abbrDR estimate, it is required that Agent~2 also transmit $\M$. In this section, we describe how a message containing $(yM,RM,M)$ should be encoded for maximum communication reduction with the \abbrGEVO method. This message coding is inspired by the previous work in \cite{Forsling2020Fusion}. A quantitative analysis of the communication reduction is provided at the end.

\subsection{Message Coding} \label{sec:dr:message-coding}

Let $\yM\in\realsm$ and assume that $\RM\in\pdsetm$ is diagonal. Suppose now that Agent~2 is transmitting $(\yM,\RM)$ to Agent~1, who fuses $(\yM,\RM)$ with its local estimate. To be able to use $(\yM,\RM)$ Agent~1 also requires $\M$, but to simply transmit $(\yM,\RM,\M)$ is costly. Instead, it is more efficient to transmit $(\yM,\Phi)$, where
\begin{equation}
	\Phi = \BBM \phi_1 \\ \vdots \\ \phi_m \EBM = \BBM r_1\psi_1  \\ \vdots \\ r_m\psi_m \EBM = \RM\M,
\end{equation}
with $r_i=[\RM]_{ii}$, and where $\phi_i$ and $\psi_i$ represent the \ith row of $\MR$ and $\M$, respectively. When Agent~1 receives $(\yM,\Phi)$, $\RM$ and $\M$ are computed as
\begin{align}
	\RM &= \diag\left(\|\phi_1\|,\hdots,\|\phi_m\|\right), & \M &= \RM\inv\Phi.
\end{align} 

The number of exchanged parameters can be further reduced by exploiting the fact that $\psi_i\psi_j\trnsp=\delta_{ij}$, and hence $\phi_i\phi_j\trnsp=0$ if $i\neq j$. For $\Phi\in\realsmnb$, the agent needs to transmit $m$ components of $\phi_1$, but only $m-i+1$ components of $\phi_i$. The components not transmitted are given by utilizing $\phi_i\phi_j\trnsp=0$ for $i\neq j$ sequentially. For example, assume $m=3$ and let $\phi_{i,j}=[\phi_i]_j$. Assume that Agent~1 does not have access to $\phi_{2,1}$, $\phi_{3,1}$, and $\phi_{3,2}$. To find the missing components, we first solve for $\phi_{2,1}$ in $\phi_1\phi_2\trnsp=0$, which is equivalent to solving $A_2\phi_{2,1} = b_2$, where
\begin{align}
	A_2 &= \phi_{1,1}, & b_2 &= -\sum_{j=2}^{\nb}\phi_{1,j}\phi_{2,j}. \nonumber
\end{align}
Then $\phi_{3,1}$ and $\phi_{3,2}$ are computed by solving $A_3\BBM\phi_{3,1}\\ \phi_{3,2}\EBM = b_3$, where
\begin{align*}
	A_3 &= \BBM\phi_{1,1} & \phi_{1,2} \\ \phi_{2,1} & \phi_{2,2} \EBM, & b_3 &= -\BBM \sum_{j=3}^{\nb}\phi_{1,j}\phi_{3,j} \\ \sum_{j=3}^{\nb}\phi_{2,j}\phi_{3,j} \EBM. 
\end{align*}
This sequential approach generalizes to arbitrary $i$. Let superscripts $M_i$ and $K_i$, \eg, $a^{M_i}$ and $a^{K_i}$, denote a vector that contains the components of a vector, \eg, $a$, corresponding to the missing and available components of $\phi_i$, respectively. In particular, the union and intersection of the components of $\phi_i^{K_i}$ and $\phi_i^{M_i}$ correspond to $\phi_i$ and the empty set, respectively. For arbitrary $i\in\{2,\dots,\nb\}$, the missing components $\phi_i^{M_i}$ are computed by solving
\begin{align}
	A_i(\phi_i^{M_i})\trnsp = b_i, \,\,
	A_i = \BBM \phi_1^{M_i} \\ \vdots \\ \phi_{i-1}^{M_i} \EBM, \,\,
	b_i = -\BBM \phi_1^{K_i}(\phi_i^{K_i})\trnsp \\ \vdots \\ \phi_{i-1}^{K_i}(\phi_i^{K_i})\trnsp \EBM.
	\label{eq:sys-of-eq}
\end{align}	
Let $\Phi^K$ denote the parts of $\Phi$ that are actually transmitted using the proposed message encoding, \ie, 
\begin{equation}
	\Phi^K = \BBM \phi_1&\phi_2^{K_2}&\hdots&\phi_m^{K_m} \EBM.
\end{equation}

Care must be taken to ensure that $\det(A_i)\neq0$, since otherwise the system of equations in \eqref{eq:sys-of-eq} is not solvable. The condition $\det(A_i)\neq 0$ has to be checked by the transmitting agent. This means that it is not possible by default to always skip transmitting the first components of $\phi_i$ since this could potentially imply that $\det(A_i)=0$. Hence, alongside the transmission of the information contained in $\phi_i$, \ie, $\phi_i^{K_i}$, the agent must also transmit information about which components are not transmitted, \ie, which elements of $\phi_i$ are included in $\phi_i^{M_i}$. This requires the transmission of a few extra parameters. For example, if $m=3$, then $1+2=3$ extra parameters must be transmitted that indicate which components of $\phi_1$ and $\phi_2$ are not transmitted. For general $m$, a number of 
\begin{equation}
	\nexcl = \sum_{i=1}^m (i-1) = \frac{m(m-1)}{2},
	\label{eq:nexcl}
\end{equation}
extra parameters must be included in the message. However, as described below, to enumerate the excluded components of the $\phi_i$, where $i=1,\dots,m$, only requires a few extra bits, which is negligible compared to the remaining parameters to be exchanged.

In summary, the transmitted message is given by $(\yM,\Phi^K,\calJ)$, where $\calJ$ is an $\nexcl$-dimensional vector containing the indices of the missing components of $\phi_i$, where $i=1,\dots,m$.\footnote{\matlab code for message coding is available at \url{https://github.com/robinforsling/dtt}.}

% COMMUNICATION CONSIDERATIONS
\subsection{Communication Reduction} \label{sec:dr:communication-reduction}

The total number of components required to be transmitted when exchanging $(\yM,\Phi^K)$ is given by
\begin{align}
	\ndr 
	&= m + \frac{\nb(\nb+1)}{2} - \frac{(\nb-m)(\nb-m+1)}{2} \nonumber \\
	&= \frac{2m\nb-m^2+3m}{2},	
	\label{eq:dr:number-of-tx-parameters}
\end{align}
where $m$ is due to $\yM$ and the remains are due to $\MR^K$. The ratio between $\ndr$ and the total number of components $\nfull=n_2(n_2+3)/2$ of a full estimate $(y_2,R_2)$ is
\begin{equation}
	\frac{\ndr}{\nfull} = \frac{2m\nb-m^2+3m}{\nb(\nb+3)} = \frac{m}{\nb}\left(2 - \frac{m+3}{\nb+3}\right)	.
	\label{eq:comm-red}
\end{equation}
The communication reduction as a function of $\nb$ when using \abbrDR is illustrated in Figure~\ref{fig:dr:communication-gain} for different values of $m$. For example, if $\nb=9$, then the communication savings are approximately $81\%$ for $m=1$ and $50\%$ for $m=3$. If $\nb=15$, then the communication savings are approximately $88\%$ for $m=1$ and $67\%$ for $m=3$. As a comparison, $\ndca=2n_2$ is included, where $\ndca$ refer to the number of parameters transmitted using the \abbrDCA framework \cite{Forsling2019Fusion}. If $m=2$, then $\ndr=2\nb+1=\ndca+1$. Hence, when comparing the performance of the two communication reduction techniques, it is reasonable to compare \abbrDCA with \abbrDR in the case of $m=2$.

\begin{figure}[tb]
	\centering
	\begin{tikzpicture}[xscale=.75,yscale=.8]
		\input{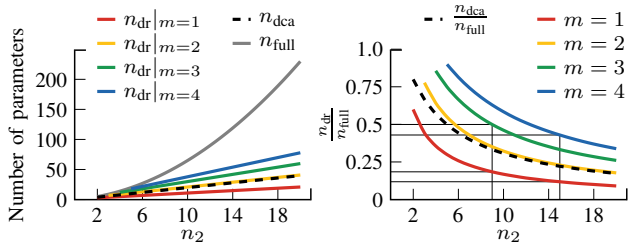}
	\end{tikzpicture}
	\caption{Communication reduction as a function of $n_2$ when using \abbrDR for different values of $m$. \abbrDCA is included for comparison. }
	\label{fig:dr:communication-gain}
\end{figure}

To continue the discussion about the $\nexcl$ extra parameters to be exchanged, assume that the size of each parameter in $(\yM,\Phi^K)$ is 32 bits. Moreover, assume that the size of each parameter in $\calJ$ is 4 bits, which is valid for $n_2\leq16$. Then $(\yM,\Phi^K)$ and $\calJ$ consist of $32\ndr$ bits and $4\nexcl$ bits, respectively. In Table~\ref{tab:extra-bits}, the ratio
\begin{equation}
	\frac{4\nexcl}{32\ndr} = \frac{m-1}{8(2n_2-m+3)},
\end{equation}
is illustrated for a number of values of $m$ and $n_2$. The amount of extra bits required to transmit $\calJ$ is marginal in this configuration. The conclusion is that the bit size of $(\yM,\Phi^K)$ and the bit size of $(\yM,\Phi^K,\calJ)$ are approximately equal.

\begin{table}[tb]
	\centering
	\caption{Percentage of Extra Bits Required for $\calJ$}
	\label{tab:extra-bits}	
	\begin{footnotesize}
	\begin{minipage}{.32\columnwidth}
	\begin{tabular}{ccc}
	\toprule\midrule
	$m$ & $\nb$ & $\frac{\nexcl}{8\ndr}$ \\	
	\midrule
		2 & 4 & 1.39\% \\ 
		2 & 6 & 0.96\% \\ 
		3 & 6 & 2.08\% \\ 
	\bottomrule
	\end{tabular}
	\end{minipage}
	\begin{minipage}{.32\columnwidth}
	\begin{tabular}{ccc}
	\toprule\midrule
	$m$ & $\nb$ & $\frac{\nexcl}{8\ndr}$ \\	
	\midrule
		3 & 9 & 1.39\% \\ 
		5 & 9 & 3.12\% \\ 
		7 & 9 & 5.36\% \\ 
	\bottomrule
	\end{tabular}
	\end{minipage}
	\begin{minipage}{.32\columnwidth}
	\begin{tabular}{ccc}
	\toprule\midrule
	$m$ & $\nb$ & $\frac{\nexcl}{8\ndr}$ \\	
	\midrule
		3 & 12 & 1.04\% \\ 
		6 & 12 & 2.98\% \\ 
		9 & 15 & 4.17\% \\ 
	\bottomrule
	\end{tabular}
	\end{minipage}			
	\end{footnotesize}
\end{table}

% --- EVALUATION ---
\section{Numerical Evaluation} \label{sec:evaluation}

The motivation for using the \abbrGEVO framework instead of, \eg, \abbrPCO is its ability to reduce dimensionality for optimal fusion performance. In this section, we evaluate \abbrGEVO numerically. We start with a simple fusion example that is parametrized by the correlation between the estimates to be fused. Then, \abbrGEVO is evaluated in a more realistic decentralized target tracking (\abbrDTT) scenario where we also compare the performance with \abbrPCO and \abbrDCA\footnote{\matlab code for the evaluation is available at \url{https://github.com/robinforsling/dtt}.}.

% TRACK FUSION EXAMPLE
\subsection{Parametrized Fusion Example}

%The \abbrGEVO method is now evaluated in a simple fusion example. 

% SIMULATION SPECIFICATIONS
\subsubsection{Simulation Specifications} 

Assume that $\nagent=2$, $\nx=\nb=6$, and $H_2=I$. The problem is to fuse $(y_1,R_1)$ and $(\yM,\RM)$. The scenario is parametrized in $\rho\in[0,1]$ according to
\begin{small}
\begin{align*}
	R_1\inv &= (1-\rho) A\inv + \rho\Gamma\inv, & 
	A\inv &= \diag\left(64,32,16,8,4,2\right), \\
	R_2\inv &= (1-\rho) B\inv + \rho\Gamma\inv, &
	B\inv &= \diag\left(5,8,13,21,34,55\right), 
\end{align*}	
\end{small}
with
\begin{align*}
	\Gamma\inv &= \BBSM 16&4&4&0&-2&0\\4&20&8&-8&-4&-4\\4&8&30&0&-4&-4\\0&-8&0&50&0&0\\-2&-4&-4&0&10&0\\0&-4&-4&0&0&20 \EBSM.
\end{align*}
The quantity $\rho\Gamma\inv$ is interpreted as common information. By construction, $R_{12}=\rho R_1\Gamma\inv R_2$. The parameter $\rho$ is varied in the interval $[0,1]$, and $(\xhat,P)$ is computed for each $\rho$ and for each of the considered methods. As $\rho$ increases, $y_1$ and $y_2$ become more correlated, eventually becoming fully correlated at $\rho=1$.

The following methods are compared:
\begin{itemize}
	\item \abbrKF: The Kalman fuser as defined in \eqref{eq:method:kf-dr}, where $\M$ is derived using Algorithm~\ref{alg:gevo-kf}.
	\item \abbrCI: Covariance intersection as defined in \eqref{eq:method:ci-dr}, where $\M$ is derived using Algorithm~\ref{alg:gevo-ci}.
	\item \abbrLE: The largest ellipsoid method as defined in \eqref{eq:method:le-dr:estimate}, where $\M$ is derived using Algorithm~\ref{alg:gevo-le}.
\end{itemize}
As was pointed out earlier, applying \abbrKF in decentralized SNs requires some sort of decorrelation mechanism. Otherwise, the independence assumed in \abbrKF is violated. In the next comparison, \abbrKF is based on two different assumptions: (i) $y_2$ is decorrelated by the removal of common information such that $R_2\inv=(1-\rho)B\inv$. This case is denoted \emph{decorrelated \abbrKF} (dKF). (ii) $y_1$ and $y_2$ are uncorrelated. This case is denoted \abbrNKF. Table~\ref{tab:parametrized-fusion-example} provides a summary of how different quantities are computed in the simulations.

\begin{table}[tb]
	\centering
	\caption{Computation of Simulated Quantities}
	\label{tab:parametrized-fusion-example}	
	\begin{footnotesize}
	\begin{tabular}{ccccc}
	\toprule\midrule
	\textbf{Method} & $R_2\inv$ & $R_{12}$ & $\M$ & $P$ \\	
	\midrule
	dKF & {$(1-\rho)B\inv$} & $0$ & Algorithm~\ref{alg:gevo-kf} & \eqref{eq:method:kf-dr} \\
	\abbrCI & {$(1-\rho)B\inv+\rho\Gamma\inv$} & $\rho R_1\Gamma\inv R_2$ & Algorithm~\ref{alg:gevo-ci} & \eqref{eq:method:ci-dr} \\
	\abbrLE & {$(1-\rho)B\inv+\rho\Gamma\inv$} & $\rho R_1\Gamma\inv R_2$ & Algorithm~\ref{alg:gevo-le} & \eqref{eq:method:le-dr:estimate} \\
	\abbrNKF & {$(1-\rho)B\inv+\rho\Gamma\inv$} & $\rho R_1\Gamma\inv R_2$ & Algorithm~\ref{alg:gevo-kf} & \eqref{eq:method:kf-dr}  \\
	\bottomrule
	\end{tabular}
	\end{footnotesize}
\end{table}

% MEASURES
\subsubsection{Evaluation Measures}

An estimate is conservative \textiff $P\succeq\Ptilde$, where $\Ptilde=\EV(\xtilde\xtilde\trnsp)$ and $\xtilde=\xhat-x$. Let $P=LL\trnsp$, where $L$ is invertible since $P\succ 0$. Since $P\succeq\Ptilde$ is equivalent to $I\succeq L\inv\Ptilde L\invtrnsp$, an estimate is conservative \textiff $\lambdamax(L\inv\Ptilde L\invtrnsp)\leq 1$. We define the \emph{conservativeness index} (\abbrCOIN) according to 
\begin{equation}
	\text{COIN} = \lambdamax\left(L\inv\Ptilde L\invtrnsp\right).
	\label{eq:coin}
\end{equation}
A common measure for conservativeness is the \emph{average normalized estimation error squared} (\abbrANEES, \cite{Li2006Fusion}). Since we have access to both $P$ and $\Ptilde$, \abbrANEES can be computed as
\begin{equation}
	\text{ANEES} = \frac{1}{\nx} \trace(P\inv\Ptilde). % \frac{1}{\nx} \trace(L\inv\Ptilde L\invtrnsp)
	\label{eq:anees}
\end{equation}

To evaluate performance, the \emph{root mean trace ratio} (\abbrRMTR) is used, which is defined as follows. Assume a fixed $m$ and let $P$ be the covariance computed by a certain method defined above. Let $\Pbc$ be the covariance computed using \eqref{eq:method:bsc-dr} with $\M$ derived using Algorithm~\ref{alg:gevo} where $R_{12}$ is known. Then
\begin{equation}
	\text{\abbrRMTR} = \frac{\sqrt{\trace(P)}}{\sqrt{\trace(\Pbc)}}.
\end{equation}

% RESULTS
\subsubsection{Results}

\begin{figure}[tb]
	\centering
	\begin{tikzpicture}[xscale=6,yscale=1.05]
		\input{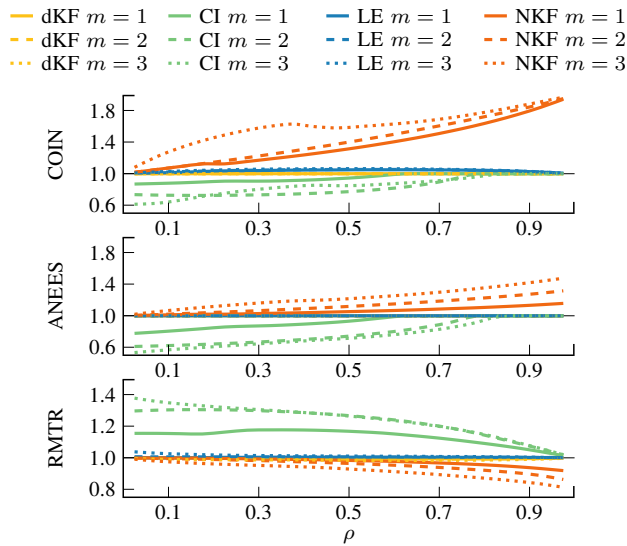}
	\end{tikzpicture}
	\caption{Results from the parametrized fusion example. Solid, dashed, and dotted lines refer to $m=1$, $m=2$, and $m=3$, respectively.}
	\label{fig:dr:gevo:evaluation:fusion-example}
\end{figure}

The results are visualized in Figure~\ref{fig:dr:gevo:evaluation:fusion-example}. Solid, dashed, and dotted lines refer to $m=1$, $m=2$, and $m=3$, respectively.

dKF is able to achieve \abbrRMTR slightly below 1. This is possible because dKF, using the ideal decorrelation step, exploits more structure in the problem compared to what is possible using the \abbrMSE optimal estimator in \eqref{eq:method:bsc-dr}. dKF is conservative \wrt both \abbrCOIN and \abbrANEES. On the other hand, \abbrNKF provides \abbrRMTR$\leq1$ but at the cost of not being conservative for $\rho>0$. \abbrNKF hence utilizes more information than is actually available. This effect becomes more prominent for large $m$ and $\rho$.

The relative performance of \abbrCI increases as $\rho$ increases, but for small $\rho$, the performance is quite poor. The reason for this is that since \abbrCI is conservative for all possible $R_{12}$, \abbrCI implicitly assumes strong correlations, but small $\rho$ means weak correlations. \abbrCOIN is typically larger for large $m$ in the case of \abbrCI.

The \abbrLE method provides relatively good performance for all $\rho$. \abbrLE is conservative \wrt \abbrANEES but not \wrt \abbrCOIN, for which it provides values marginally larger than 1.

% DISCUSSION
%\subsubsection{Discussion}
%
%The previous example shows that each of the proposed \abbrGEVO methods is relevant if both performance and conservativeness criteria are considered. The comparison also provides an indication under which conditions each method is favorable. If decorrelation is possible before fusion, then \GEVOKF performs well. If it is not possible to decorrelate the estimates, then \naive usage of \GEVOKF will typically result in a non-conservative estimate. If the cross-correlations are unknown but relatively small or some level of non-conservativeness is allowed, then \GEVOLE is a relevant option. If the cross-correlations are unknown and possible strong, or if conservativeness is not negotiable, then \GEVOCI provides the most robust solution.

% DTT EVALUATION
\subsection{Decentralized Target Tracking Example}

%The \abbrGEVO method is evaluated using a \emph{decentralized target tracking} (\abbrDTT) scenario. \abbrGEVO is also compared to \abbrPCO.

% SIMULATION SPECIFICATIONS
\subsubsection{Simulation Specifications}

\begin{figure}[tb]
	\centering	
	\begin{tikzpicture}[scale=.25]
		\input{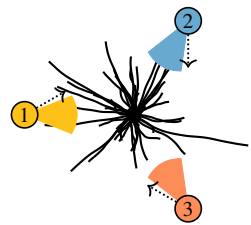}
	\end{tikzpicture}
	\caption{The \abbrDTT scenario. The target is initially located at the black circle. Three agents (colored numbered circles) estimate the target and exchange local estimates. The arrows defines the communication topology. Gray lines represent target trajectories in different \abbrMC runs.}
	\label{fig:evaluation:dtt:scenario}
\end{figure}

The \abbrDTT scenario consists of $N=3$ agents and one dynamic target in $d=2$ spatial dimensions. A \emph{Monte Carlo} (\abbrMC) approach is considered, where $M$ denotes the number of \abbrMC runs. The scenario is illustrated in Fig.~\ref{fig:evaluation:dtt:scenario}. A \emph{constant velocity model} (\abbrCVM, \cite{Li2003TAES}) is assumed for the target such that $x\in\reals^4$. In each \abbrMC simulation, the target initial state $x_0$ is sampled from $\calN(0,P_0)$, where $P_0$ is predetermined. The target then evolves according to the \abbrCVM
\begin{align}
	x_{k+1} &= F_kx_k + w_k, & w_k &\sim \calN(0,Q_k),
\end{align}
for a predetermined $Q_k$. As a consequence, the trajectory is unique in each \abbrMC simulation. A measurement $z_{i,k}$ of the \ith agent at time $k$ model as
\begin{align}
	z_{i,k} &= \BBM 1&0&0&0 \\ 0&1&0&0 \EBM x_k + e_{i,k}, & e_{i,k}&\sim\calN(0,C_{i}),
\end{align}
where $e_{i,k}$ is the measurement noise at time $k$ and $C_{i}\in\pdset^2$ the measurement noise covariance of Agent~$i$. A Kalman filter is used for state estimation. \abbrBSC, \abbrKF, \abbrCI, and \abbrLE are used for fusion. The $Q_k$ used for sampling the target trajectory is the same $Q_k$ used in all local filters. Since the initial state, the process noise, and the measurement are all Gaussian distributed, the \ith local estimate at time $k$ is on the form
\begin{align}
	y_{i,k} &= x_k + v_{i,k}, & v_{i,k}&\sim\calN(0,R_{i,k}).
\end{align}
Hence, in this case, the estimation error $\xtilde_k$ is a zero-mean Gaussian distributed random variable. Relevant simulation parameters are summarized in Table~\ref{tab:evaluation:dtt:simulation-parameters}. The communication is scheduled such that Agent~$i$ transmits \textiff $k\in\{i,i+N,i+2N,\dots\}$. The communication topology is given by the arrows in Fig.~\ref{fig:evaluation:dtt:scenario}. It is assumed that $H_2=I$ and $m=2$.

\begin{table}[tb]
	\centering
	\caption{Parameters used in the simulations}
	\label{tab:evaluation:dtt:simulation-parameters}
	\begin{footnotesize}
	\begin{tabular}{cl}
		\toprule\midrule
		\textbf{Parameter} & \textbf{Comment} \\
		\midrule
		$m=2$ & dimensionality of the \abbrDR estimate $\yM$ \\
		$d=2$ & spatial dimensionality \\
		$\nx=4$ & state dimensionality \\
		$T_s=1$ & sampling time [\second] \\
		$\sigma_{w}=2$ & standard deviation of process noise [\meter\second$^{-\frac{3}{2}}$] \\		
		$C_1=\BBSM100&0\\0&10\EBSM$ & measurement noise covariance of Agent~1 [\meter$^2$] \\
		$C_2=\BBSM33&39\\39&78\EBSM$ & measurement noise covariance of Agent~2 [\meter$^2$] \\
		$C_3=\BBSM33&-39\\-39&78\EBSM$ & measurement noise covariance of Agent~3 [\meter$^2$] \\
		%$\nk=15$ & number of time steps \\
		$M=10\,000$ & number of \abbrMC runs \\
		\bottomrule
	\end{tabular}	
	\end{footnotesize}
\end{table}

The following methods are compared:
\begin{itemize}
	\item \abbrNKF: The \naive Kalman fuser as defined in \eqref{eq:method:kf-dr}, where $\M$ is derived using Algorithm~\ref{alg:gevo-kf}.
	\item \abbrCI: Covariance intersection as defined in \eqref{eq:method:ci-dr}, where $\M$ is derived using Algorithm~\ref{alg:gevo-ci}.
	\item \abbrLE: The largest ellipsoid method as defined in \eqref{eq:method:le-dr:estimate}, where $\M$ is derived using Algorithm~\ref{alg:gevo-le}.
	\item DCA-EIG: Covariance intersection as defined in \eqref{eq:method:ci-dr} with $\M=I$. Agent~$i$ transmits $(y_i,\Ds_i)$, where $\Ds_i$ is computed according to Algorithm~\ref{alg:dca-eig}.
\end{itemize}

% MEASURES
\subsubsection{Evaluation Measures} 

\abbrCOIN and \abbrANEES are used to evaluate conservativeness. However, in this \abbrMC evaluation these are defined slightly differently.

Let $\xhat_k$ denote an estimate of $x_k$ after fusion in some agent. In this case, we do not have access to $\Ptilde_k=\EV(\xtilde_k\xtilde_k\trnsp)$, where $\xtilde_k=\xhat_k-x_k$. However, $\Ptilde_k$ can be approximated by
\begin{equation}
	\Phat_k = \frac{1}{M} \sum_{j=1}^M \xtilde_k^j(\xtilde_k^j)\trnsp,
	\label{eq:Phat}
\end{equation}
where $\xtilde_k^j$ is the error in the \jth \abbrMC run. By construction $P_k^j=P_k=L_kL_k\trnsp$ is the same over all \abbrMC runs. This leads to
\begin{equation}
	\text{\abbrCOIN}_k = \lambdamax\left(L_k\inv\Phat_k L_k\invtrnsp\right).
	\label{eq:coin-approx}
\end{equation}
\abbrANEES is now defined as
\begin{equation}
	\text{\abbrANEES}_k = \frac{1}{\nx} \trace(P_k\inv \Phat_k).
	\label{eq:anees-approx}
\end{equation}

To evaluate performance, we use the \abbrRMTR which in this case is defined as
\begin{equation}
	\text{\abbrRMTR}_k = \frac{\sqrt{\trace(P_k)} \text{ using \abbrDR estimates}}{\sqrt{\trace(P_k)}  \text{ using full estimates}}.
\end{equation}
By construction, \abbrRMTR$\geq 1$.

% RESULTS
\subsubsection{Results}

Fig.~\ref{fig:evaluation:dtt:results} illustrates the results for Agent~3 only\footnote{The estimation results for the other two agents are similar.}. \abbrCOIN, \abbrANEES, and \abbrRMTR are computed at time instants where fusion is performed. As a comparison, the corresponding results for the \abbrPCO method have been included for \abbrRMTR. The \abbrPCO curves are given by the dotted lines and use the same color coding as the corresponding \abbrGEVO curves, \eg, the red dashed lines refer to \abbrNKF with $\M$ derived using the \abbrPCO method.

\begin{figure}[tb]
	\centering
	\begin{tikzpicture}[xscale=.45,yscale=.75]
		\input{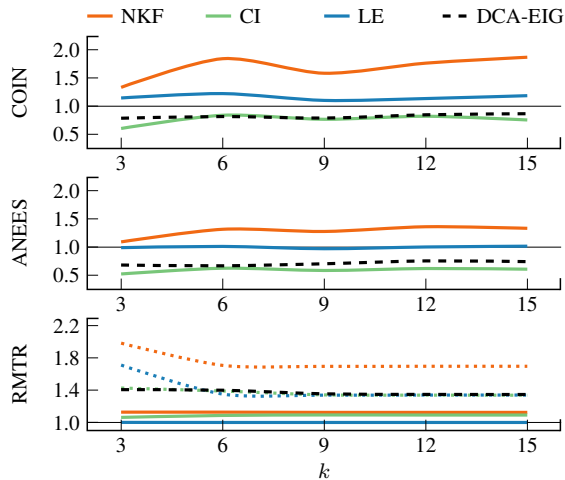}
	\end{tikzpicture}
	\caption{Results from the \abbrGEVO method evaluation. The dotted curves illustrate the results when using the \abbrPCO method.}
	\label{fig:evaluation:dtt:results}
\end{figure}

\abbrCI is the only method that is conservative \wrt both \abbrCOIN and \abbrANEES. \abbrLE is conservative \wrt \abbrANEES. \abbrNKF is never conservative. As indicated by the \abbrRMTR plot, \abbrNKF, \abbrCI, and \abbrLE achieves significantly better performance when using \abbrGEVO instead of \abbrPCO. For \abbrLE, the performance loss of transmitting $(\yM,\RM)$ instead of $(y_2,R_2)$ is zero. This can be explained by the binary behavior suggested by \abbrLE, \cf \eqref{eq:method:le-dr:transformed-estimate}, where each component of the fused estimate is taken exclusively from either of the two estimates that are fused. In this case, it is useless to increase $m$ since the additional information in $(\yM,\RM)$ cannot be utilized by \abbrLE. It can also be seen that \abbrGEVO outperforms DCA-EIG.

%An interesting observation is that \abbrRMSER is smaller than 1 for \abbrNKF. The reason is that when full estimates are exchanged, the double counting of information by the \naive \abbrKF method makes the \abbrRMSE diverge. This effect is reduced when the \abbrGEVO derived \abbrDR estimates are exchanged instead of full estimates.

% --- CONCLUSIONS ---
\section{Conclusions} \label{sec:conclusions}

Decentralized state estimation under communication constraints has been considered. The problem involved agents sharing information in a decentralized sensor network, where only \abbrDR estimates were allowed to be communicated. In particular, we have proposed the \emph{generalized eigenvalue optimization} (\abbrGEVO) method as a framework for computing \abbrDR estimates optimized for fusion performance. Essentially, the \abbrGEVO method computes the matrix $\M$ which defines the \abbrDR estimate to be communicated.

The \abbrGEVO framework was based on a generalization of the \abbrBSC formulas. The main focus has been on three variants of \abbrGEVO corresponding to the following fusion methods: \abbrKF, \abbrCI, and the \abbrLE method. We have shown how the optimal linear mapping is computed for \abbrBSC, \abbrKF, and \abbrLE. For \abbrCI, an alternating minimization algorithm has been proposed, which was shown to converge to a stationary point. 

Different aspects related to the \abbrGEVO method were discussed. For instance, dealing with singularities when deriving $\M$ and how to assign the number of rows in $\M$. A message-coding solution was provided for efficient communication when using \abbrDR estimates. We have also compared the \abbrGEVO framework, both theoretically and numerically, with two baseline methods. The numerical examples demonstrated the usability of the \abbrGEVO framework and its superiority to the baseline methods.

%Moreover, it shows that the different proposed \abbrGEVO methods are preferable in different problems, and none of the proposed \abbrGEVO methods is generally superior the other ones.

%In this paper it has been assumed that $R_1$ is available for deriving $\M$. In real-world problems this is in general an unrealistic assumption, particularly in arbitrary decentralized sensor networks. For future research it is therefore important to investigate how $R_1$ can be replaced without a significant performance degradation.

%\input{sec/appendix.tex}

% BIBLIOGRAPHY
\bibliographystyle{IEEEtran.bst}
\bibliography{bib/myrefs}

% BIOGRAPHY

\end{document}